\documentclass{article}
\usepackage{amsthm}
\usepackage{color}
\usepackage{graphicx}

\newcommand{\bigo}{{\mathrm{O}}}
\newcommand{\height}{{\mathrm{height}}}
\newcommand{\ymin}{{\mathrm{ymin}}}
\newcommand{\ymax}{{\mathrm{ymax}}}
\newcommand{\hmax}{{\mathrm{hmax}}}
\newcommand{\m}{{\mathrm{m}}}
\newcommand{\M}{{\mathrm{M}}}
\newcommand{\val}{{\mathrm{v}}}
\newcommand{\action}{{\mathrm{vtype}}}
\newcommand{\set}{{\mathrm{on}}}
\newcommand{\low}{{\mathrm{above}}}
\newcommand{\rai}{{\mathrm{below}}}
\newcommand{\dist}{{\mathrm{dist}}}

\newtheorem{theorem}{Theorem}
\newtheorem{lemma}{Lemma}
\newtheorem{definition}{Definition}

\begin{document}



\title{PLANAR SHAPE MANIPULATION USING APPROXIMATE GEOMETRIC PRIMITIVES}

\author{Victor Milenkovic\thanks{Department of Computer Science, University of Miami.
Coral Gables, FL 33124-4245, USA.
vjm@cs.miami.edu}
and
Elisha Sacks\thanks{Computer Science Department, Purdue University.
West Lafayette, IN 47907-2066, USA.
eps@cs.purdue.edu}
and
Steven Trac\thanks{Department of Computer Science, University of Miami.
Coral Gables, FL 33124-4245, USA.
strac@cs.miami.edu}}

\maketitle

\begin{abstract}
We present robust algorithms for set operations and Euclidean transformations of
curved shapes in the plane using approximate geometric primitives.  We use a
refinement algorithm to ensure consistency.  Its computational complexity is
$\bigo(n\log n+k)$ for an input of size $n$ with $k=\bigo(n^2)$ consistency
violations.  The output is as accurate as the geometric primitives.  We validate
our algorithms in floating point using sequences of six set operations and
Euclidean transforms on shapes bounded by curves of algebraic degree~1 to~6.  We
test generic and degenerate inputs.

\end{abstract}

\section{Introduction}

Set operations and Euclidean transformations are important computational
geometry tasks with many applications.  There are efficient algorithms for
planar shapes, which are more common and simpler than 3D shapes.  A shape is
modeled as a subdivision: a partition of the plane into faces, curves, and
points.  Fig.~\ref{f-region}a shows a subdivision with faces $f_0$ and $f_1$,
curves $e_1,\ldots,e_4$, and points indicated with dots.  A set operation is
performed by constructing the mutual refinement of the input subdivisions,
called the {\em overlay,} and returning the faces that satisfy the set operator
(Fig.~\ref{f-region}c).  A Euclidean transformation is performed by transforming
the points, curves, and faces of the input subdivision.

\begin{figure}[tbp]
\centering
\begin{tabular}{ccc}
\raisebox{0.1875in}{\includegraphics{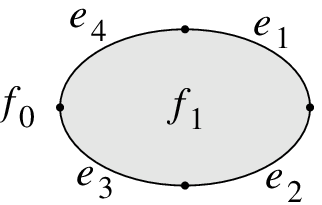}} & 
\includegraphics{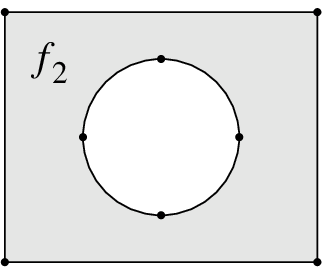} &
\includegraphics{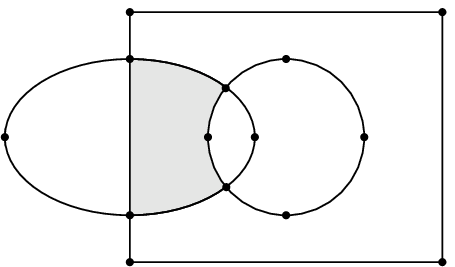}\\
(a) & (b) & (c)
\end{tabular}
\caption{Subdivisions (a, b) and their overlay (c) with $f_1\cap f_2$ shaded.}
\label{f-region}
\end{figure}

Shape manipulation algorithms use geometric primitives that are formulated in
the real-RAM model where real arithmetic is exact and has unit cost.  The {\em
  robustness problem\/} is how to implement these primitives in computer
arithmetic.  The mainstream strategy is to implement primitives exactly using
integer arithmetic and algebraic computation.  We prefer approximate primitives
in floating point because they are much faster and have constant bit complexity.
Although the approximation error is negligible, it can cause prior shape
manipulation algorithms to generate inconsistent or highly inaccurate outputs.
We present shape manipulation algorithms that use approximate primitives yet
generate consistent, accurate outputs.

\subsection{Prior work}

Yap\cite{yap04} describes exact computational geometry.  Exact computation
increases bit complexity, hence running time.  The CGAL library\cite{cgal}
provides floating point filtering techniques that somewhat reduce this cost.
The increase in bit complexity from input to output leads to unbounded
complexity in sequences of computations.  The only current solution is to
simplify the output.  Although simplification algorithms have been developed for
a few domains,\cite{m-spr97,hp-isr-02,eigenwillig07,fortune97} not including
implicit algebraic curves, a general strategy is unknown.

The alternative to exact primitives is approximate primitives.  The robustness
problem is to ensure consistency despite approximation error.

The \emph{controlled perturbation} robustness strategy is to verify that the
predicates in an algorithm are correct despite approximation error.
Verification fails when the predicate value is too close to zero.  Failure is
prevented probabilistically by randomly perturbing the input parameters.  If any
predicate still cannot be verified, the input is perturbed again. Controlled
perturbation has been applied to
arrangements,\cite{r-cpaps-99,hs-pssaamm-98,halperin04a} convex
hulls,\cite{funke05} and Delaunay triangulation.\cite{funke05} The perturbation
size is exponential in the algebraic degree of the primitives.  The degree-four
perturbation size exceeds the error bounds of engineering applications.  We
developed versions of controlled perturbation that solve this problem and used
them to compute Minkowski sums of polyhedra\cite{sacks-milenkovic11} and free
spaces of planar robots.\cite{sacks-milenkovic12} Controlled perturbation does
not work for shape manipulation because it does not handle algebraic curves and
does not address output simplification in sequences of computations.

Our robustness strategy is to eliminate inconsistency due to approximate
geometric primitives by enforcing consistency constraints on the data structures
that represent geometric objects.  We analyze the running time and the accuracy
in terms of the approximation error, $\delta$, and the number of constraint
violations, $k$, but neither parameter is provided to the algorithms.  An
algorithm is {\em inconsistency sensitive\/} when the extra running time for
inconsistency elimination is polynomial in $k$ and $n$, the extra error is
polynomial in $k$ and $n$ and is linear in $\delta$, and $k$ is polynomial in
$n$.  Inconsistency sensitivity captures the concept of an efficient approximate
algorithm: fast and accurate when there are few inconsistencies, and degrading
gracefully.

We developed an inconsistency sensitive arrangement algorithm for
algebraic-curve segments.\cite{sacks-milenkovic06} The $x$-coordinates where
segments cross and the segment $y$-order between crossings are computed with
$(1+kn)\epsilon$ accuracy for $n$ segments, where $\epsilon$ denotes the
error in the approximate geometric primitives of that algorithm.  A consistency
constraint violation is a cyclic $y$-order, due to an incorrect ordering of
segment crossing coordinates.  The algorithm computes a consistent $y$-order of
size $V=2n+N+\min(3kn,n^2/2)$ for $N$ crossings.  The running time is
$\bigo(V\log n)$ and the output is correct for a set of segments within
$(1+kn)\epsilon$ of the output segments.  We also developed two inconsistency
sensitive Minkowski sum algorithms.\cite{sacks-milenkovic10}

The arrangement algorithm can compute a curve $y$-order that is inconsistent
with the curve endpoints, as illustrated below.  In prior work, we developed a
simple solution for the special case where every curve endpoint has a distinct
$x$-coordinate.  We handled the general case heuristically by sweeping along a
random axis.  This strategy proved inadequate even for generic inputs.  In this
paper, we present a general, inconsistency sensitive solution.

\subsection{Inconsistency in the Overlay Algorithm}

Our overlay algorithm uses three geometric primitives: computing intersection
points of monotone curves, computing turning points of non-monotone curves, and
comparing coordinates of points.  The first primitive is required by any overlay
algorithm, while the other two are dictated by our use of a sweep algorithm.
The high level algorithm is as follows.  Split the curves at their turning
points.  Split the monotone curves at their intersection points to obtain
sub-curves.  Derive the $y$-order of each pair of sub-curves that overlap in $x$
by comparing the $y$ coordinates of their intersection points with a vertical
line at the middle of their interval of $x$-overlap.  The details of the sweep
and how it avoids creating cycles appears in previous
work\cite{sacks-milenkovic10} and in Sec.~\ref{s-set}.

The algorithm output, the $x$-order of the curve endpoints and the partial
$y$-order of the sub-curves, allows us to compute the subdivision structure
without additional approximate primitives.  In particular, we can use a standard
sweep algorithm to generate a vertical cell decomposition.  The curve endpoints
are the sweep events in order of increasing $x$.  The sweep list contains curves
in $y$-order whose left endpoint event has been processed and whose right
endpoint has not.  If a subdivision construction algorithm, e.g.
randomized incremental vertical cell decomposition, requires endpoint/curve
order, the $y$-order of endpoint $a$ of curve $ab$ with respect to curve $cd$,
with $c_x<a_x<d_x$, is deduced from the calculated $y$-order of $ab$ and $cd$.

We illustrate that an inconsistent $y$-order can arise in computing the overlay
of a monotone curve, $e_2$, with two monotone curves, $e_1$ and $e_3$, that meet
at a point, $b$ (Fig.~\ref{f-bad}a).  Due to numerical error, the intersection
primitive incorrectly reports that $e_2$ and $e_3$ do not intersect.
Intersections with vertical lines (arrows) order $e_1$ above $e_2$ and $e_2$
above $e_3$.  The incorrect $y$-order of $e_2$ and $e_3$ is correct for a small
deformation of the input (Fig.~\ref{f-bad}b).  But the algorithm output is
inconsistent because no deformation can make $e_2$ be below $e_1$, above $e_3$,
and disjoint from $b$.

\begin{figure}[tbp]
\centering
\begin{tabular}{ccc}
\includegraphics{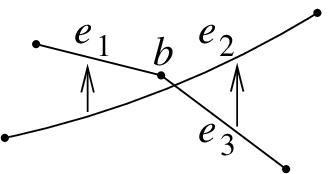} & \includegraphics{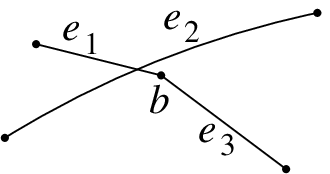} & \includegraphics{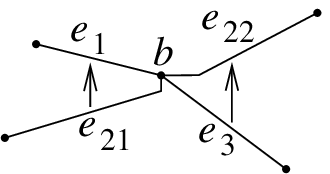}\\
(a) & (b) & (c)
\end{tabular}
\caption{Inconsistent $y$-order (a), $e_2$ deformation (b), and consistent
  refinement (c).}\label{f-bad}
\end{figure}

\subsection{Contribution and Organization}

Our main result is an inconsistency sensitive \emph{refinement algorithm} that
assigns a consistent $y$-order to a set of curves by splitting some curves at
endpoints of other curves (Fig.~\ref{f-bad}c).  The output $y$-order is the
induced refinement of the input $y$-order.  In our example, $e_2$ is split into
$e_{21}$ and $e_{22}$ at $b$, $e_{21}$ is below $e_1$, $e_{22}$ is above $e_3$,
and the other pairs are unordered because their domains are disjoint.  The
precondition is that each {\em pair} of input curves can be placed in the input
$y$-order using a deformation of at most $\delta$.  The postcondition is that a
{\em single} $\delta$-deformation can place {\em all} the output curves in the
output $y$-order.  The running time is $\bigo(n\log n + k)$ for $n$ curves with
$k=\bigo(n^2)$ consistency violations.  We define concepts in
Sec.~\ref{s-concepts}, present the refinement algorithm in Sec.~\ref{s-refine},
and analyze it in Secs.~\ref{s-correct}--\ref{s-error}.

As indicated above, overlay can introduce inconsistencies. Surprisingly, even
transformations can do so.  To accomplish set operations, we start with a
possibly inconsistent overlay, run the refinement algorithm, construct the
faces, and extract the relevant faces (Sec.~\ref{s-set}).  To transform a
subdivision, we transform its curves, split rotated curves at new turning
points, run refinement, and construct the faces (Sec.~\ref{s-transform}). In
this manner, the refinement algorithm enables us to accomplish consistent and
accurate shape manipulation using approximate geometric primitives.

We present an empirical validation that our algorithms are fast and accurate on
sequences of six set operations and Euclidean transformations on shapes bounded
by curves of algebraic degree~1 to~6 (Sec.~\ref{s-validate}).  Moreover, $k$ is
zero for generic input and is small for degenerate input.  We conclude with a
discussion (Sec.~\ref{s-conclude}).

\section{Concepts}\label{s-concepts}

In this section, we define approximate subdivisions and four ways they can be
inconsistent.  Fig.~\ref{f-bad}a is an example of an approximate subdivision and
the depicted inconsistency is an instance of one of these ways.  The refinement
algorithm eliminates these inconsistencies.  We prove that the output represents
a consistent shape (Sec.~\ref{s-correct}) and that its error is bounded by the
input error (Sec.~\ref{s-error}).  Each group of related concepts is defined
informally then formally.

We work in the $xy$ plane.  A point, $a$, has Cartesian coordinates $(a_x,a_y)$.
Curves are open, bounded, and piecewise algebraic.  They are assumed
$x$-monotone and $y$-monotone unless stated otherwise.  A curve, $e$, whose
endpoints are $a$ and $b$ with $a_x\leq b_x$ is denoted $e[a,b]$.  It is
vertical when $a_x=b_x$, horizontal when $a_y=b_y$, increasing when $a_y<b_y$,
and decreasing when $a_y>b_y$.

We represent a subdivision as a set of disjoint curves
(Def.~\ref{def:subdivision}).  The combinatorial structure of the subdivision is
derivable from the $x$-order of the curve endpoints and the $<_y$ partial order
on the curves (Def.~\ref{def:yorder}), as explained in the introduction.

\begin{definition}[subdivision]\label{def:subdivision}
A {\em subdivision\/} is a set of curves in which every member is disjoint from
the closure of every other member.
\end{definition}

\begin{definition}
Sets $s$ and $t$ {\em overlap in $x$\/} if there exist $p \in s$ and $q \in t$
with $p_x = q_x$.
\end{definition}

\begin{definition}[$<_y$]\label{def:yorder}
For two point sets, $s$ and $t$, that overlap in $x$, $s<_y t$ if $p_y < q_y$
for each $p \in s$ and $q \in t$ with $p_x = q_x$.  For a point, $a$, $a <_y s$
means $\{a\} <_y s$.
\end{definition}

We model a shape with an {\em approximate subdivision}: a set of curves with an
imposed acyclic partial $y$-order (Def.~\ref{def:appsub}) denoted $\prec_y$.  An
approximate subdivision is a subdivision when its curves are disjoint from the
closures of other curves and the $\prec_y$ order matches the $<_y$ order.

\begin{definition}[approximate subdivision]\label{def:appsub}
An {\em approximate subdivision\/} is a set of curves with a binary relation,
denoted $\prec_y$, whose transitive closure is irreflexive, such that
$e\prec_yf$ or $f\prec_ye$ iff $e$ and $f$ overlap in $x$.  The notation
$e\preceq_y f$ means $e\prec_y f$ or $e=f$, $e \succ_y f$ means $f \prec_y e$,
and $e \succeq_y f$ means $f \preceq_y e$.
\end{definition}

Fig.~\ref{f-rep}a shows an example in which the $\prec_y$ order for curves that
overlap in $x$ is index order, e.g. $e_3\prec_ye_5$ because $3<5$.  Although
$\prec_y$ agrees with $<_y$ here, they disagree elsewhere.  For example, $e_2$
intersects $e_5$, while $e_9\prec_ye_{10}$ yet $e_{10}<_ye_9$.

\begin{figure}[tbp]
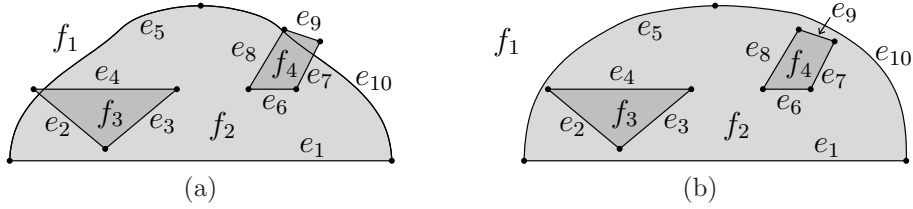

\centering
\begin{tabular}{c@{\hspace{0.5in}}c}
\input{rep2L.pstex_t} & \input{rep1L.pstex_t} \\
(a) & (b)
\end{tabular}
\caption{Approximate subdivision (a) and realization (b).}
\label{f-rep}
\end{figure}

An approximate subdivision (Fig.~\ref{f-rep}a) has a \emph{realization} when an
endpoint-preserving deformation (Fig.~\ref{f-rep}b) of the curves can make the
actual $<_y$ order agree with the imposed $\prec_y$ order
(Def.~\ref{def:realization}).  A realizable approximate subdivision represents
an infinite number of shapes, but their curves all have the same endpoints and
$<_y$ order, hence all these shapes have the same combinatorial structure.  An
unrealizable approximate subdivision does not represent any shape.

\begin{definition}[realization]\label{def:realization}
A {\em realization\/} of an approximate subdivision, $S$, is a function, $r$,
from $S$ onto a subdivision, $T$, such that for all $e\in S$, $e$ and $r(e)$
have the same endpoints and for all $e,f\in S$, $e\prec_yf$ implies
$r(e)<_yr(f)$.
\end{definition}

As indicated in the introduction, the vertical order of an endpoint, $p$, of a
curve, $f$, with respect to a curve $e$, can be deduced from the order of $e$
and $f$.  In this manner, we define the lower and upper sets of endpoints,
$L(e)$ and $U(e)$, with respect to $e$ (Def.~\ref{def:LU}), from the imposed
vertical order $\prec_y$.  We use these sets to define four consistency
constraints (Def.~\ref{def:consistent}).  Our central theorem is that an
approximate subdivision has a realization iff it satisfies these four
consistency constraints (Thm.~\ref{t-iff}).

\begin{definition}[$L(e)$ and $U(e)$]\label{def:LU}
For a curve, $e[a,b]\in S$, an approximate subdivision, $L(e)$ is the set of
endpoints, $p$, of curves, $f\preceq_y e$, with $a_x\leq p_x\leq b_x$; likewise,
$U(e)$ with $f\succeq_y e$.
\end{definition}

\begin{definition}[consistent]\label{def:consistent}
An approximate subdivision, $S$, is {\em consistent\/} if each curve, $e[a,b]\in
S$, satisfies four constraints.
\begin{enumerate}
\item For each $q\in U(e)$, there is no $p\in L(e)$ with $a_x<p_x=q_x<b_x$
  and $p_y\geq q_y$.
\item For each $f[c,d]\succ_y e$, $a_x=c_x$ implies $a_y\leq
  c_y$ and $b_x=d_x$ implies $b_y\leq d_y$.
\item a) For each $p \in L(e)$, there are no $c,d\in U(e)$ with
  $c_x<p_x<d_x$ and $c_y,d_y\leq p_y$. b) For each $q \in U(e)$, there are no
  $c,d\in L(e)$ with $c_x<q_x<d_x$ and $c_y,d_y\geq q_y$.
\item If $e$ is vertical, there is no endpoint, $p$, of curve $f$, with $p_x=a_x$ and
  $p_y\in(a_y,b_y)$, the open interval between $a_y$ and $b_y$.
\end{enumerate}
\end{definition}

Fig.~\ref{f-inconsistent} illustrates the constraints.  Part~(a) violates
constraint~1 if $e\prec_yg$, $f\prec_ye$, $p_x=q_x$, and $p_y\geq q_y$.
(Fig.~\ref{f-bad}a is a special case with $e=e_2$, $f=e_3$, $g=e_1$ and
$p=q=b$.)  Part~(b) violates constraint~2 if $e\prec_yf$, $a_x=c_x$, and
$a_y>c_y$.  Part~(c) violates constraint~3a if $f\prec_ye$, $e\prec_yg$,
$e\prec_yh$, $c_x<p_x<d_x$, and $p_y>c_y,d_y$.  Part~(d) violates constraint~4
if $a_x=p_x=b_x$ and $a_y<p_y<b_y$.  All four of these violations occur in
practice (Sec.~\ref{s-validate}).

\begin{figure}[tbp]
\centering \includegraphics[width=4.5in]{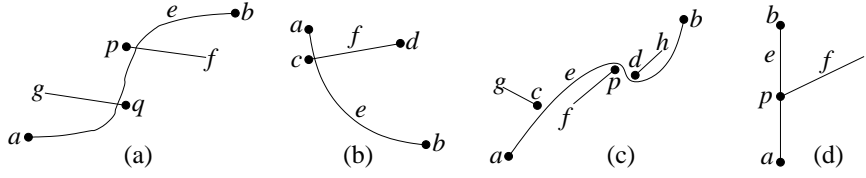}
\caption{Inconsistent approximate subdivisions.}\label{f-inconsistent}
\end{figure}

\setcounter{theorem}{2}
\begin{theorem}
An approximate subdivision is realizable iff it is consistent (p.~\pageref{t-iff}).
\end{theorem}

A \emph{refinement} of an approximate subdivision is an approximate subdivision
whose curves are sub-curves of the original curves with the induced $\prec_y$
order (Def.~\ref{def:refappsub}).

\begin{definition}[refinement of curve]
A {\em refinement\/} of a curve, $e[a,b]$, is a set of curves, $C(e)$, that form
an $x$-monotone chain from $a$ to $b$.
\end{definition}

\begin{definition}[refinement of approximate subdivision]\label{def:refappsub}
A {\em refinement\/} of an approximate subdivision, $S$, is an approximate
subdivision, $S'$, with the same set of endpoints in which each curve, $e\in S$,
maps to a refinement, $C(e)\subset S'$, and $e\prec_yf$ in $S$ implies
$e'\prec_yf'$ in $S'$ for every $e'\in C(e)$ and $f'\in C(f)$ that overlap in
$x$.
\end{definition}

Fig.~\ref{f-refinement} shows consistent refinements of the inconsistent
approximate subdivisions from Fig.~\ref{f-inconsistent} with
$C(e)=\{e_1,e_2,e_3\}$ in parts~(a) and~(c) and $C(e)=\{e_1,e_2\}$ in parts~(b)
and~(d).  The refinements are realizations, except in part~(c) where $e_1<_yf$
but $f\prec_y e$, and $h<_ye_3$ but $e\prec_y h$.  However, a realization exists
by Thm.~\ref{t-iff}.

\begin{figure}[tbp]
\centering
\includegraphics[width=4.5in]{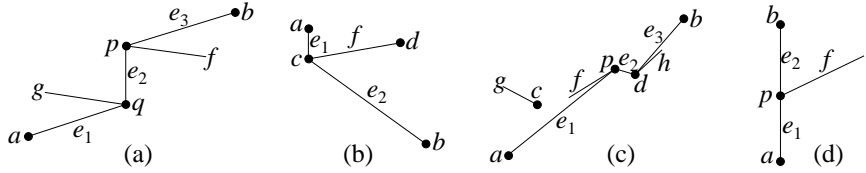}
\caption{Refinements of the approximate subdivisions in Fig.~\ref{f-inconsistent}.}
\label{f-refinement}
\end{figure}

\setcounter{theorem}{0}
\begin{theorem}
The refinement algorithm computes a consistent refinement (p.~\pageref{t-refinealg}).
\end{theorem}

The remaining concepts relate to error analysis.  We define an error metric on
approximate subdivisions that models errors in the pairwise $\prec_y$ orders due
to errors in the geometric primitives used to compute them, especially curve
intersection point computation.  An approximate subdivision is
\emph{$\delta$-accurate} when the $\prec_y$ order of each pair of curves equals
the $<_y$ order of a pair of $\delta$-close curves
(Def.~\ref{def:delta-accurate}).  Each pair of curves can be placed in $\prec_y$
order by a small deformation.

\begin{definition}
A {\em $\delta$-deformation\/} of a curve, $e$, is an $x$-monotone curve with
the same endpoints as $e$ whose Hausdorff distance from $e$ is less than
$\delta$.
\end{definition}

\begin{definition}[$\delta$-accurate]\label{def:delta-accurate}
An approximate subdivision is {\em $\delta$-accurate\/} if $e\prec_yf$ implies
that $e$ and $f$ have $\delta$-deformations, $\tilde{e}$ and $\tilde{f}$, with
$\tilde{e}<_y\tilde{f}$.
\end{definition}

Fig.~\ref{f-delta-accurate}a--b show that Fig.~\ref{f-inconsistent}a is
$\delta$-accurate for a small $\delta$, which matches our intuition that the
curves barely intersect.  The pairs $e,f$ and $e,g$ use different $\tilde{e}$
curves.  Fig.~\ref{f-delta-accurate}c shows that Fig.~\ref{f-inconsistent}b is
$\delta$-accurate.  Even though $c\in\tilde{e}$, $\tilde{e}<_y\tilde{f}$ because
curves are open.  Fig.~\ref{f-delta-accurate}d shows that
Fig.~\ref{f-inconsistent}c is $\delta$-accurate.  The curves are their own
$\delta$-perturbations, and $\tilde{e}$ is $x$-monotone but not monotone.
Fig.~\ref{f-inconsistent}d is $\delta$-accurate because $e$ and $f$ do not
overlap in $x$.

\begin{figure}[tbp]
\centering
\includegraphics[width=4.5in]{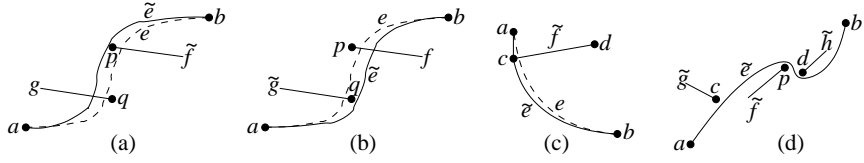}
\caption{Illustrations of $\delta$-accuracy.}\label{f-delta-accurate}
\end{figure}

A refinement is a \emph{$\delta$-splitting} when every split point is
$\delta$-close to the curves that it splits (Def.~\ref{def:accemb}).  Although
the curves in Fig.~\ref{f-refinement} are inaccurate, the refinement is a
$\delta$-splitting because $e$ is split at points that are $\delta$-close to it.

\begin{definition}[$\delta$-splitting]\label{def:accemb}
A refinement of an approximate subdivision, $S$, is a {\em $\delta$-splitting\/}
if for all $e\in S$ and $e'[a',b']\in C(e)$, $\dist(a',e),\dist(b',e)<\delta$,
where $\dist(a,e)$ denotes the Hausdorff distance from a point, $a$, to a curve,
$e$.
\end{definition}

\setcounter{lemma}{14}
\begin{lemma}
The refinement algorithm output is a $\delta$-splitting (p.~\pageref{l-refine}).
\end{lemma}
\setcounter{lemma}{0}

The output error of the refinement algorithm is due to the input error and the
refinement error.  If the input is $\delta$-accurate and the refinement is a
$\delta$-splitting, the output is $\delta$-accurate.  A stronger error bound is
that the output has a realization in which each curve is $\delta$-close to its
preimage curve, called a \emph{$\delta$-refinement} (Def.~\ref{def:d-emb}).  A
$\delta$-refinement has a single $\delta$-deformation that places all the curves
in $\prec_y$ order, whereas $\delta$-accuracy allows a separate
$\delta$-deformation for each pair of curves.  The realizations in
Fig.~\ref{f-delta-refinement} show that the refinements in
Fig.~\ref{f-refinement} are $\delta$-refinements.

\begin{definition}[$\delta$-refinement]\label{def:d-emb}
A {\em $\delta$-refinement\/} of an approximate subdivision, $S$, is a
refinement that has a realization such that the Hausdorff distance from $r(e')$
to $e$ is bounded by $\delta$ for all $e\in S$ and $e'\in C(e)$.
\end{definition}

\begin{figure}[tbp]
\centering
\includegraphics[width=4.5in]{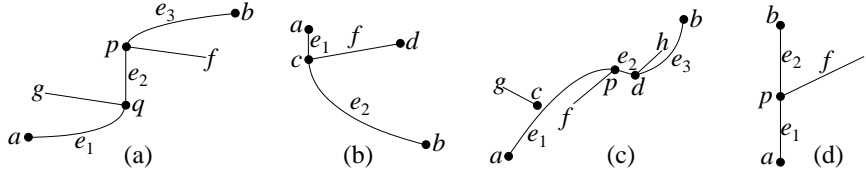}
\caption{Realizations of Fig.~\ref{f-refinement} that are $\delta$-close to the
  approximate subdivisions in
  Fig.~\ref{f-inconsistent}.}\label{f-delta-refinement}
\end{figure}

We conclude the error analysis by proving (Sec.~\ref{s-error}) that the stronger
error bound is implied by the weaker one and consistency, hence that the
refinement algorithm computes a $\delta$-refinement of its input.

\setcounter{theorem}{4}
\begin{theorem}
A consistent $\delta$-splitting of a $\delta$-accurate approximate subdivision
is a $\delta$-refinement (p.~\pageref{t-refine}).
\end{theorem}

\begin{theorem}
The refinement algorithm computes a $\delta$-refinement (p.~\pageref{t-delta}).
\end{theorem}
\setcounter{theorem}{0}

\section{Refinement Algorithm}\label{s-refine}

The refinement algorithm consists of four steps that enforce the four
constraints in Def.~\ref{def:consistent} by splitting curves at endpoints of
other curves.  Splitting $e[t,h]$ at $p$ replaces $e$ by a refinement,
$\{e_1[t,p],e_2[p,h]\}$, such that $e_i\prec_yf$ (or $e_i\succ_yf$) iff
$e\prec_yf$ (or $e\succ_yf$) and $e_i$ and $f$ overlap in $x$.  For simplicity,
one can choose $e_1[t,p]$ and $e_2[p,h]$ to be the line segments $tp$ and $ph$.
Even though this may be inaccurate, the ultimate realization of the refinement
will be close to the original curves (Thm.~\ref{t-refine}).

In the following description, a curve, $e[t,h]$, is classified with respect to
an endpoint coordinate, $x$, as follows: $t_x=x=h_x$, vertical; $t_x=x<h_x$,
outgoing; $t_x<x=h_x$, incoming; and $t_x<x<h_x$, passing.

\paragraph{Step 1}

We begin with an example using curves $f_1$ to $f_{10}$ and $\prec_y$ equal to
index order for curves that overlap in $x$ (Fig.~\ref{f-through}a).  Because
$f_7\prec_yf_9$, $w\in L(f_9)$, yet because $f_9\prec_yf_{10}$, $w\in U(f_9)$.
Hence, curve $f_9$ and endpoint $w$ violate constraint~1 if we set $e=f_9$,
$p=w$, and $q=w$.  The algorithm visits, in an order consistent with $\prec_y$,
the curves that have an endpoint whose $x$-coordinate equals $w_x$: $f_1$,
$f_2$, $f_3$, $f_4$, $f_5$, $f_7$, $f_8$, and $f_{10}$.  (Curve $f_4[u,v]$ is
vertical and $u_x=v_x=w_x$.)  It forms groups of these curves separated by
passing curves: $G_1=\{f_1,f_2,f_3,f_4,f_5\}$, passing curve $f_6$,
$G_2=\{f_7,f_8\}$, passing curve $f_9$, and $G_3=\{f_{10}\}$.  Each curve is
added to the latest group and the constraint is checked.  The check fails when
the highest endpoint, $p$, in the previous group and the lowest endpoint, $q$,
in the current group satisfy $p_y\geq q_y$.  The algorithm splits at $p$ each
passing curve $e$ that separate the previous group from the current group
(because $e$, $p$, and $q$ violate constraint~1), and it then merges these two
groups and the split curves into a single group.  In our example, the check
fails after $f_{10}$ is added to $G_3$, $f_9$ is split at $w$ into $f'_9$ and
$f''_9$, $G_3$ is merged into $G_2$, and the algorithm is done
(Fig.~\ref{f-through}b).

\begin{figure}[tbp]
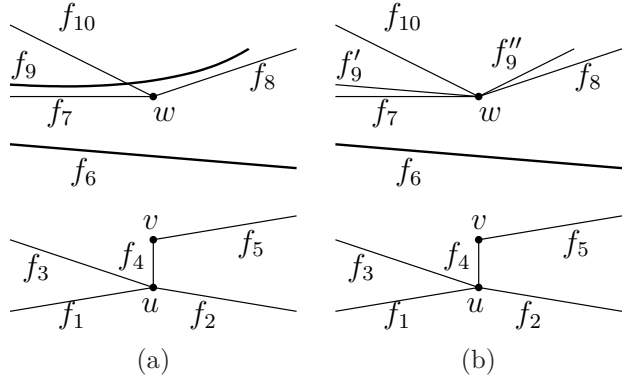

\centering
\begin{tabular}{cc}
\input{through1L.pstex_t} & \input{through2L.pstex_t}\\
(a) & (b)
\end{tabular}
\caption{Curves before (a) and after (b) constraint~1 enforcement.}
\label{f-through}
\end{figure}

The algorithm is a plane sweep with a vertical sweep line.  The sweep list, $L$,
consists of curves in $\prec_y$ order.  The events are the endpoint
$x$-coordinates in increasing order.  The incoming curves at $x$ are removed
from $L$, the non-passing curves are handled in $\prec_y$ order as described
below, and the outgoing curves are inserted into $L$.  We write $f\prec_y G$ or
$G\prec_y f$ when $f\prec_y g$ or $g\prec_y f$ for every $g$ in group $G$.  The
endpoints, $v$, of curves in $G$ with $v_x=x$ and with minimum and maximum
$y$-coordinates are denoted $\ymin(G)$ and $\ymax(G)$.  A curve, $e$, is handled
as follows.

\begin{flushleft}
1. If $m=0$ or there exists a passing $f$ with  $G_m\prec_y f\prec_ye$\\
\hspace{2.25em}    increment $m$ and set $G_m=\{e\}$\\
\hspace{0.9em}    else add $e$ to $G_m$.\\
2. While $m>1$ and $\ymin(G_m)_y\leq\ymax(G_{m-1})_y$:\\
\hspace{2.25em} a. For each passing curve, $G_{m-1}\prec_y g\prec_y G_m$:\\
\hspace{4.5em} i. Remove $g$ from $L$.\\ 
\hspace{4.5em} ii. Split $g$ into $g'$ and $g''$ at $\ymax(G_{m-1})$.\\
\hspace{4.5em} iii. Add $g'$ and $g''$ to $G_{m-1}$.\\
\hspace{2.25em} b. Insert the curves of $G_m$ into $G_{m-1}$ and decrement $m$.
\end{flushleft}
\mbox{}\vspace{-2.5em} 

\paragraph{Step 2}

We begin with an example using curves $f_1$ to $f_4$ that are incoming at
$x={b_1}_x={b_2}_x={b_3}_x={b_4}_x$ and with $\prec_y$ equal to index order
(Fig.~\ref{f-incoming}a).  The pair $\left<f_2,f_3\right>$ violates constraint~2
with $e=f_2$, $f=f_3$, $b=b_2$ and $d=b_3$, and similarly
$\left<f_1,f_3\right>$.  We can handle $\left<f_2,f_3\right>$ by splitting $f_2$
at $b_3$ or by splitting $f_3$ at $b_2$.  The latter is more accurate because
$\dist(b_2,f_3)$ (shown with a dashed line) is less than $\dist(b_3,f_2)$.  We
say that $b_2$ is safe for $\left<f_2,f_3\right>$.  Likewise, $b_3$ is safe for
$\left<f_1,f_3\right>$ if $\dist(b_3,f_1)\leq\dist(b_1,f_3)$.  We remove the two
violations by splitting $f_1$ into $\{f'_1, b_3b_1\}$ and $f_3$ into
$\{f'_3,b_2b_3\}$ (Fig.~\ref{f-incoming}b).

\begin{figure}[tbp]
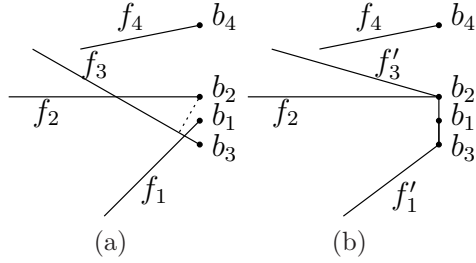

\centering
\begin{tabular}{cc}
\input{incoming1L.pstex_t} & \input{incoming2L.pstex_t}\\
(a) & (b)
\end{tabular}
\caption{Curves before (a) and after (b) constraint~2 enforcement.}
\label{f-incoming}
\end{figure}

The algorithm enforces constraint~2 via safe splits of curves that violate it.
It processes the incoming curves at each endpoint $x$-coordinate as follows;
outgoing curves are processed analogously.  Let $s$ be the list of incoming
curves in $\prec_y$ order.  A safe endpoint for $e[v,a]\in s$ is an endpoint,
$b$, of a curve, $f[w,b]\in s$, with $\dist(b,e)\leq\dist(a,f)$.  All distances
are measured with respect to the input curves: for a refinement curve $e'\in
C(e)$ (Def.~\ref{def:refappsub}), use $\dist(b,e)$ in place of $\dist(b,e')$.
For endpoints $a$ and $b$, let $\min(a,b)$ denote $a$ if $a_y<b_y$ and $b$
otherwise; likewise $\max(a,b)$.  Define $\m(e)$ as the minimum endpoint over
$g\succeq_ye$ that is safe for $e$.  In our example, $\m(f_1)=b_3,
\m(f_2)=b_2,\m(f_3)=b_3,\m(f_4)=b_4$.  The algorithm computes $\m(e)$ for every
$e\in s$ then splits each curve, $e$, at the maximum $\m(f)$ over $f\preceq_y
e$.

\begin{flushleft}
1.  For each curve, $e[v,a]$, in $s$:\\
\hspace{2.25em}a. Set $\m(e)=a$.\\
\hspace{2.25em}b. Append $e$ to $I$.\\
\hspace{2.25em}c. While $a_y<b_y$ for the predecessor, $f[w,b]$, of $e$ in $I$:\\
\hspace{4.75em}i.  Swap $e$ and $f$ in $I$.\\
\hspace{4.75em}ii. If $a$ is safe for $f$, set $\m(f)=\min(\m(f), a)$.\\
2. For each curve, $e$, in $s$\\
\hspace{2.25em} Split $e$ at the maximum of $\m(f)$ over $f\preceq_ye$ in $s$.
\end{flushleft}

Steps~1--2 of the algorithm compute $\m(e)$ by inserting each curve, $e[v,a]$,
into an initially empty list, $I$, ordered by endpoint $y$-coordinate.  First,
$e$ is appended to $I$ then it is swapped with each predecessor, $f[w,b]$, with
$b_y>a_y$ until it reaches the correct position in $I$.  In our example, $f_3$
is inserted into $I=(f_1,f_2)$ with two swaps after which $I=(f_3,f_1,f_2)$.
Each swap reveals a constraint violation, hence implies a possible update to
$\m(f_3)$.

\paragraph{Step 3}

We begin with an example with curves $g[t,h]$, $f_1[a_1,b_1]$ and $f_2[a_2,b_2]$
with $f_1\prec_yg$ and $g\prec_yf_2$ (Fig.~\ref{f-monotone}a).  Curve $e=g$ and
$p=a_2\in U(e)$ violate constraint~3b with $c=t$, and $d=b_1$.  The algorithm
constructs a monotone curve from $t$ to $h$ that is above every $p\in L(g)$ and
is below every $q\in U(g)$.  It visits these endpoints in increasing $x$ order
and tracks the rightmost one, $\val(g)$, where the curve is forced to increase
or decrease.  A constraint violation is detected when $\val(g)$ switches between
increasing and decreasing, and $g$ is split at the previous $\val(g)$.
Specifically, $\val(g)$ starts as $t$, is set to $a_2$, is unchanged at $b_2$
and $a_1$, is set to $b_1$, a violation is detected, and $g$ is split at $a_2$
(Fig.~\ref{f-monotone}b).

\begin{figure}[tbp]
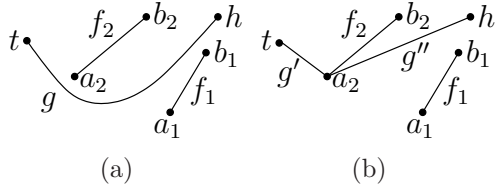

\centering
\begin{tabular}{cc}
\input{monotone1L.pstex_t} &
\input{monotone2L.pstex_t} \\
(a) & (b)
\end{tabular}
\caption{Curves before (a) and after (b) constraint~3 enforcement.}
\label{f-monotone}
\end{figure}

The algorithm is a sweep with $L$ as before.  Each $e\in L$ stores an endpoint,
$\val(e)$, and a flag, $\action(e)$, that indicates if $\val(e)$ is on, above,
or below $e$.  Each input curve, $f[t,h]$, has a tail event at $t_x$ and a head
event at $h_x$.  The events are handled in increasing $x$ order with ties broken
as follows: heads of non-vertical curves, tails of vertical curves, heads of
vertical curves, and tails of non-vertical curves.

\newpage 
\begin{flushleft}
tail event for $f[t,h]$:\\ 
1. If $f$ is not vertical, insert it into $L$ with $\val(f)=t$ and 
   $\action(f)=\set$.\\
2. For each $e\in L$ with $e\prec_y f$ and $t_y\leq\val(e)_y$, call 
   $\mathrm{lower}(e,t)$.\\ 
3. For each $g\in L$ with $g\succ_y f$ and $t_y\geq\val(g)_y$, call 
   $\mathrm{raise}(g,t)$.\\[1em]

head event for $f[t,h]$:\\
1.  For each $e\in L$ with $e\preceq_y f$ and $h_y\leq\val(e)_y$, 
    call $\mathrm{lower}(e,h)$.\\
2.  For each $g\in L$ with $g\succeq_y f$ and $h_y\geq\val(g)_y$, 
    call $\mathrm{raise}(g,h)$.\\
3.  If $f$ is not vertical, remove it from $L$.\\[1em]

subroutine $\mathrm{lower}(e,a)$:\\
1. If $\action(e)=\rai$, split $e$ at $\val(e)$.\\
2. Set $\val(e)=a$ and $\action(e)=\low$.\\[1em]

subroutine $\mathrm{raise}(e,a)$:\\
1. If $\action(e)=\low$, split $e$ at $\val(e)$.\\
2. Set $\val(e)=a$ and $\action(e)=\rai$.
\end{flushleft}

This algorithm handles strict extrema.  In the general case, lower sets a
pointer from $a$ to $\val(e)$ if $\action(e)=\low$ and $\val(e)_y=a_y$; likewise
if $\action(e)=\rai$ in raise.  The split steps traverse the pointers and split
$e$ at the extrema.

\paragraph{Step 4}

Split every vertical curve at every endpoint that violates constraint~4.

\section{Correctness}\label{s-correct}

We prove that the refinement algorithm output is a consistent refinement of the
input (Thm.~\ref{t-refinealg}) and that the algorithm is inconsistency sensitive
(Thm.~\ref{t-time}).  We then prove that an approximate subdivision is
consistent iff it has a realization (Thm.~\ref{t-iff}), which implies that the
refinement algorithm output has a realization.

\begin{theorem}\label{t-refinealg}
The refinement algorithm computes a consistent refinement.
\end{theorem}
\begin{proof}
We prove consistency.  Being a refinement follows from the definition of a curve
split.

Step~1 enforces the invariant that constraint~1 holds for every passing curve,
$f$, and the endpoints in $G_1,\ldots,G_m$.  Step~2 enforces constraint~2 by
assigning each incoming curve an endpoint whose $y$-coordinate is the maximum
over all incoming curves below it in $\prec_y$ order.  For step~3, let input
curve $e$ and endpoint $q\in U(e)$ violate constraint~3b with $c,d\in L(e)$; the
other case is similar.  If $e$ is split at an endpoint, $w$, with $c_x<w_x<d_x$
before $d$ is handled, we are done.  Otherwise, $\val(e)_y\geq c_y$ after $c$ is
handled, so $\val(e)_y\leq q_y$ and $\action(e)=\low$ after $q$ is handled.
When $d$ is handled, $\action(e)=\low$ and $\val(e)=r$ for some $r\in U(e)$ with
$r_y\leq q_y$ and $c_x<r_x< d_x$, since no split occurs before $d$.  Thus, $e$
is split at $r$, which eliminates the $e,q$ inconsistency.  Step~4 trivially
enforces constraint~4.

Step~2 maintains constraint~1 by not changing $L(e)$ or $U(e)$ for a passing
curve, $e$.  Step~3 maintains constraint~1 by not adding an endpoint, $v$, to
$U(e)$, except by splitting $e$ at $v$ (proof below).  It maintains constraint~2
by splitting a curve, $e$, at the endpoint $v\in L(e)$ or $v\in U(e)$ with the
largest or smallest $v_y$ value.  Step~4 maintains constraints~1--3 by only
splitting vertical curves.

Proof: Suppose $e\prec f$, $v\in L(e)\cap L(f)$, and $f$ is split at $v$.  We
show that $e$ is also split at $v$.  After $v$ is processed, $\val(f)=v$, since
it is later split at $v$.  By Lem.~\ref{l-step3} below, $\val(e)_y\leq v_y$.
But $\val(e)_y\geq v_y$ because $v\in L(e)$, so $\val(e)=v$.  Curve $f$ is split
at $v$ when the sweep processes an endpoint, $u\in U(f)$ with $u_y\leq v_y$.  If
$\val(e)=v$ then, $e$ is split at $v$ because $u\in U(e)$.  Otherwise, $e$ is
lowered by $w\in U(e)\cap L(f)$ with $w_x<u_x$ and $e$ is split at $v$ then.
\end{proof}

\begin{lemma}\label{l-step3}
In step~3, $\val(e)_y\leq\val(f)_y$ for $e,f\in L$ with $e\prec_y f$.
\end{lemma}
\begin{proof}
The invariant holds initially and is preserved by sweep updates because if $v\in
L(e)$ overlaps $f$ in $x$, then $v\in L(f)$.  If we call $\mathrm{raise}(e,v)$ and if $v_y\geq\val(f)_y$, then we also call $\mathrm{raise}(f,v)$. Hence $\val(e)_y\leq\val(f)_y$ remains true because $\val(e)_y=\val(f)_y=v$.  Similarly for $u\in U(f)$.
\end{proof}

To prove inconsistency sensitivity, we need to count constraint violations.  A
constraint~1 violation is a pair $e,q$ such that there exists $p$ that violates
the constraint, and similarly the other violations are defined by the pair
$e,f$ (2), the pair $e,p$ (3a), the pair $e,q$ (3b), or the pair $e,p$ (4).

\begin{theorem}\label{t-time}
The refinement algorithm complexity is $\bigo(n\log n+k)$ for $n$ curves with
$k=\bigo(n^2)$ constraint violations.
\end{theorem}
\begin{proof}
The definitions of the constraint violations imply that $k=\bigo(n^2)$, since
there are $n$ curves and $\bigo(n)$ endpoints.  Each curve, $e$, stores its
poset height, $\height(e)$.  For $e$ and $f$ that overlap in $x$, $e\prec_yf$
iff $\height(e)<\height(f)$.  Each curve stores pointers to the curves below and
above its endpoints.  Splitting a curve updates these pointers in constant time.
Sort the endpoints in $x$ order with ties broken by height.
This processing takes $\bigo(n\log n)$
time and $\bigo(n)$ space.

Step~1 takes $\bigo(n+k)$ time using a doubly linked list to represent $L$.
Group operations take constant time using $\ymin$, $\ymax$, and the maximum
curve height, $\hmax$, to represent a group.

Step~2 takes $\bigo(n+k)$ time because the number of constraint~2 violations is
$\bigo(k)$ after step~1, which is shown as follows.  If step~1 splits $e$ at
$p$, $p$ is the highest point in $G_{m-1}$, so the new curves satisfy
constraint~2 with respect to a curve, $f\prec_ye$.  A constraint~2 violation is
created for each curve, $f\succ_y e$, that has an endpoint, $q$, with $q_x=p_x$
and $q_y<p_y$.  Since $q\in U(e)$, $e$ and $f$ are also a constraint~1
violation.

Step~3 takes $\bigo(n\log n+k)$ time.  Maintain a partition of $L$ into sublists
with a common $\val$ and $\action$.  Each sublist stores pointers to its lowest
and highest curves.  Represent the partition as a balanced binary tree of
sublists ordered by $\val_y$.  By Lem.~\ref{l-step3}, these will also be ordered
by $h$.  Store the curves in $L$ in a binary tree ordered by $h$.  To handle an
event, look up the new curve in the latter binary tree to determine the $h$ of
its neighbors.  Look up this $h$ in the former binary tree to determine the
sublist that will be split.  The new curve can lower the curves below it in its
sublist or can raise the curves above.  Each case splits the sublist into lower
and upper sublists.  When a lower or a raise propagates to sublists below or
above, they are merged into the newly created lower or upper sublists.  A tail
event also creates a singleton sublist for the curve that enters $L$.  A head
event removes its curve from its sublist and removes an empty sublist from the
partition.

Excluding sublist merges and curve splits, the sweep takes $\bigo(n\log n)$
because there are $\bigo(n)$ tree operations.  Although an event can cause
$\bigo(n)$ merges, the merge time is $\bigo(n\log n)$ because each merge
decreases the number of sublists, whereas at most $2n$ sublists are created.
Since each split applies to an entire sublist, the partition structure is
unchanged and the only cost is $\bigo(k)$ to update the curves.

Step~4 takes $\bigo(n+k)$ time, using the sorted endpoints, because the number
of constraint~4 violations is $\bigo(k)$ after steps~1--3.  Each split creates
at most one violation and is charged to a constraint~1--3 violation.
\end{proof}

We now turn to Thm.~\ref{t-iff}.  We first prove necessity
(Lem.~\ref{lem:onlyif}).  To prove sufficiency, we start with a more general
sufficient condition (Def.~\ref{def:ordered} and Lem.~\ref{lem:ordered}).

\begin{lemma}\label{lem:onlyif}
If an approximate subdivision has a realization, it is consistent.
\end{lemma}
\begin{proof}
For constraint~1, consider $p\in L(e)$ and $q\in U(e)$ with $a_x<p_x=q_x<b_x$
(Fig.~\ref{f-inconsistent}a).  There exists a curve $f\prec_ye$ with endpoint
$p$ and a curve $g\succ_ye$ with endpoint $q$.  Hence, $r(f)<_yr(e)$ and
$r(e)<_yr(g)$.  By continuity, $p<_y\overline{r(e)}$ or $p\in\overline{r(e)}$,
but $r(e)\cap\overline{r(f)}=\emptyset$ by Def.~\ref{def:realization} and
$a_x<p_x<b_x$, hence $p<_yr(e)$.  Similarly, $r(e)<_yq$ and hence $p_y<q_y$.

For constraint~2, consider $a_x=c_x$ (Fig.~\ref{f-inconsistent}b).  Since
$r(e)<_yr(f)$, $a\in\overline{r(e)}$, and $c\in\overline{r(f)}$, $a_y\leq c_y$.
The $b_x=d_x$ case is similar.

For constraint~3a, consider $e[a,b]$ increasing (Fig.~\ref{f-inconsistent}c).
As above, $p\in L(e)$ and $a_x<p_x<b_x$ imply $p<_yr(e)$, and $d\in U(e)$ and
$a_x<d_x$ imply $r(e)<_yd$ or $d=b$.  Since $r(e)$ is monotone and $p_x<d_x$,
$p_y<d_y$.  Similarly, $p_y<c_y$ for $e$ horizontal or decreasing.  The
constraint~3b proof is similar.

Constraint~4 holds because $r(e)\cap\overline{r(f)}=\emptyset$ by
Def.~\ref{def:realization} (Fig.~\ref{f-inconsistent}d).
\end{proof}

\begin{definition}[monotone]\label{def:M}
A set, $M(e)$, is monotone with respect to a curve, $e[a,b]$, if $M(e)=e$ for
$e$ horizontal or vertical and otherwise $M(e)$ is an open set bounded by
monotone curves, $l(e)$ and $u(e)$, that connect $a$ to $b$
(Fig.~\ref{fig:MfMe}a).
\end{definition}

\begin{definition}[ordered monotone assignment]\label{def:ordered}
An assignment of a monotone set, $M(e)$, to each curve, $e$, in an approximate
subdivision is {\em ordered\/} if no $M(e)$ contains a curve endpoint, and
$f\prec_ye$ implies $M(f)<_yM(e)$ unless $e$ and $f$ are both increasing or both
decreasing in which case $M(f)<_y u(e)$ (Fig.~\ref{fig:MfMe}b).
\end{definition}

\begin{definition}[$p^-$, $p^+$, $S^-$ and $S^+$]\label{def:p^-}
For a point, $p$, $p^- = \{ (p_x,y) | y \leq p_y \}$ and $p^+ = \{ (p_x,y) | y
\geq p_y \}$.  For a set, $S$, $S^- = \bigcup_{p \in S} p^-$ and $S^+ =
\bigcup_{p \in S} p^+$.
\end{definition}

\begin{lemma}\label{lem:ordered}
If an approximate subdivision has an ordered monotone assignment, it has a
realization.
\end{lemma}
\begin{proof}
For each curve, $e$, we will choose $r(e)\subset M(e)$ that will be disjoint
from the other curve endpoints.  If $e$ is horizontal or vertical, define
$r(e)=M(e)=e$.  Vertical $r(e)$ and $r(f)$ cannot intersect because they cannot
contain each other's endpoints.  Otherwise, $f \prec_y e$ implies $r(f)<_y r(e)$
because $M(f) <_y M(e)$, except when both are increasing or both are decreasing.

We define $m$ for the increasing curves in an order that is consistent with
$\prec_y$.  Let $M'(e)$ equal $M(e)$ minus $r(f)^-$ for all increasing $f\prec_y
e$ (Fig.~\ref{fig:MfMe}c).  Because $r(f)\subset M(f) <_y u(e)$
(Fig.~\ref{fig:MfMe}b), $M'(e)$ is connected.  Since the endpoints of $r(f)$ do
not lie in $R(e)$, the only part of the boundary of $r(f)^-$ that can lie in
$R(e)$ is $r(f)$, hence the lower boundary of $M'(e)$ is monotone.  Choose a
monotone path, $r(e)\in M'(e)$, from $a$ to $b$ (Fig.~\ref{fig:MfMe}c).  Because
$r(e)\cap r(f)^- \subset M'(e)\cap r(f)^- = \emptyset$, $r(f)<_yr(e)$ as
required.  We handle the decreasing curves similarly.
\end{proof}

\begin{figure}[tbp]
\centering \includegraphics[width=4.5in]{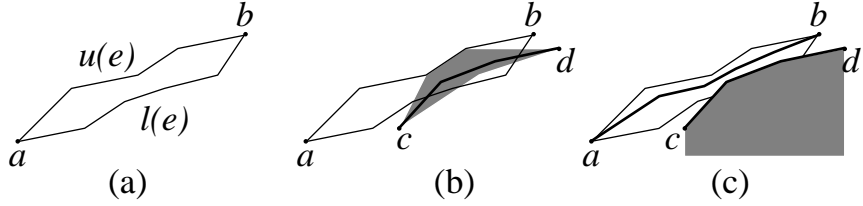}
\caption{(a) $M(e)$; (b) $M(f)$ (shaded) and $r(f)$ (dark); (c) $M'(e)$, $r(e)$
  (dark), and $r(f)^-$ (shaded).}\label{fig:MfMe}
\end{figure}

We complete the sufficiency proof by showing that a consistent approximate
subdivision has an ordered monotone assignment, $R(e)$.

\begin{definition}\label{def:a++}
For a point, $p$, $p^{++}=[p_x,\infty)\times[p_y,\infty)$,
    $p^{-+}=(-\infty,p_x]\times[p_y,\infty)$,
    $p^{+-}=[p_x,\infty)\times(-\infty,p_y]$, and
    $p^{--}=(-\infty,p_x]\times(-\infty,p_y]$.  For a set, $A$,
$A^{++}=\bigcup_{p\in A}p^{++}$ {\em etc}.  For two points, $p$ and $q$, with
$p_x<q_x$, $A(p,q)=(p_x,q_x)\times(p_y,q_y)$.
\end{definition}

\begin{definition}[$R(e)$]
For a curve $e[a,b]$:
\[
\begin{array}{cc}
R(e) = & \left\{\begin{array}{ll} A(a,b)-L(e)^{+-}\cup U(e)^{-+} & \mbox{for
  increasing $e$,} \\ A(a,b)-L(e)^{--}\cup U(e)^{++} & \mbox{for decreasing
  $e$,} \\ e & \mbox{otherwise.}
\end{array}\right.
\end{array}
\]
\end{definition}

Lemmas \ref{lem:p++} through \ref{lem:Rordered} apply to a consistent
approximate subdivision.

\begin{lemma}\label{lem:p++}
For a curve, $e[a,b]$, if $p\in L(e)$, $q\in U(e)$, and either $p\not\in\{a,b\}$
or $q\not\in\{a,b\}$, $e$ increasing or horizontal implies $p^{+-}\cap
q^{-+}=\emptyset$ and $e$ decreasing or horizontal implies $p^{--}\cap
q^{++}=\emptyset$.
\end{lemma}
\begin{proof}
Suppose $e'[a',b']$ is increasing or horizontal, $f'\prec_ye'$ has endpoint
$p'\in L(e')$, $g'\succ_ye'$ has endpoint $q'\in U(e')$, $q'\not\in\{a',b'\}$,
and ${p'}^{+-}\cap{q'}^{-+}\not=\emptyset$, hence $p'_x\leq q'_x$ and $p'_y\geq
q'_y$.  We will show a contradiction.  The other cases are similar.

If $a'_x<p'_x=q'_x<b'_x$, constraint~1 is violated with $e[a,b]=e'[a',b']$,
$p=p'$, and $q=q'$.  If $a'_x=p'_x=q'_x$, $p'_y\leq a'_y$ by constraint~2 with
$e=f'$, $f=e'$, $a=p'$, and $c=a'$.  Similarly, $a'_y\leq q'_y$, so $a'_y\leq
q'_y\leq p'_y\leq a'_y$ and $q'=a'$, which is a contradiction.  Similarly,
$p'_x=q'_x=b'_x$ is contradictory.  If $a'_x<p'_x<q'_x<b'_x$, either $p'_y\geq
a'_y$ or $q'_y\leq b'_y$.  The former contradicts constraint~3a with $e=e'$,
$c=a$, $p=p'$, and $d=q'$, and the latter contradicts constraint~3b with $e=e'$,
$c=p$, $q=q'$, and $d=b'$.  If $a'_x=p'_x<q'_x<b'_x$, $p'_y\leq a'_y$ by
constraint~2 and constraint~3b is contradicted as before.  Similarly,
$a'_x<p'_x<q'_x=b'_x$ is contradictory.  If $a'_x=p'_x$ and $q'_x=b'_x$,
$p'_y\leq a'_y$ and $b'_y\leq q'_y$ by constraint~2, so $b'_y\leq q'_y\leq
p'_y\leq a'_y\leq b'_y$ and $q'=b'$, which is contradictory.
\end{proof}

\begin{lemma}\label{lem:Rmonotone}
$R(e)$ is monotone with respect to $e$.
\end{lemma}
\begin{proof}
For increasing $e$, the boundary of $p^{+-}$ is monotone, so the boundary of
$L(e)^{+-}$ is monotone because it is the upper envelope of a set of monotone
curves.  Similarly, the boundary of $U(e)^{-+}$ is monotone.  By
Lem.~\ref{lem:p++},
\begin{displaymath}
L(e)^{+-}\cap U(e)^{-+}= \cup_{p\in L(e), q\in U(e)} p^{+-}\cap q^{-+}= \{a,b\}.
\end{displaymath}
Hence, $R(e)$ is an open set bounded below and above by the portions of the
boundaries of $L(e)^{+-}$ and $U(e)^{-+}$ that connect $a$ to $b$.  Decreasing
$e$ is handled similarly and the other cases are trivial.
\end{proof}

\begin{lemma}\label{l-trans}
For $e$, $f$, and $g$ that overlap in $x$, $e\prec_yf$ and $f\prec_yg$ implies
$e\prec_yg$.
\end{lemma}
\begin{proof}
Since $e$ and $g$ overlap in $x$, $e\prec_yg$ or $g\prec_ye$, but the latter
implies that $(e,e)$ is in the transitive closure of $\prec_y$, which
contradicts Def.~\ref{def:appsub}.
\end{proof}

\begin{lemma}\label{lem:Rordered}
The assignment of $R(e)$ to $e$ is ordered.
\end{lemma}
\begin{proof}
Verticals cannot contain endpoints by constraint~4.  A horizontal, $e[a,b]$,
cannot contain an endpoint, $p$, by constraint~3 with $c=a$ and $d=b$.
Otherwise, $R(e)$ cannot contain $p$ because one of
$p^{--},p^{-+},p^{+-},p^{++}$ has been removed.

If $e[a,b]$ is not decreasing and $f[c,d]$ is not increasing and $f \prec_y e$,
either $a=c$, $a \in U(f)-\{c,d\}$, or $c \in L(e)-\{a,b\}$.  If $a=c$, $e$ and
$f$ cannot both be horizontal without one containing the other's right endpoint.
We have $p_y\leq a_y$ for $p\in R(f)$ and $q_y\geq a_y$ for $q\in R(e)$ with
equality only for horizontal curves.  Hence, $R(f)<_yR(e)$.  If $a \in
U(f)-\{c,d\}$, $R(f)\cap a^{++}=\emptyset$.  Since $R(e) \subset a^{++}$, $R(f)
<_y R(e)$.  The $c \in L(e)-\{a,b\}$ case is similar, as is $e$ not increasing
and $f$ not decreasing.

It remains to show that if $e[a,b]$ and $f[c,d]$ are both increasing and $f
\prec_y e$, $R(f) <_y u(e)$, and similarly for both decreasing.  It suffices to
show $R(f)<_yq^{-+}$ for each $q\in U(e)$ for which $q^{-+}$ and $f$ overlap in
$x$.  If $q_x<d_x$ (Fig.~\ref{fig:Rconnected}, $q=1$), a curve, $g\succ_ye$,
with endpoint $q$ also overlaps $f$ in $x$, so $g\succ_yf$ by
Lem.~\ref{l-trans}, $q\in U(f)$, $q^{-+}$ was subtracted from $R(f)$, and thus
$R(f)<_yq^{-+}$.  If $q_x>d_x$ (Fig.~\ref{fig:Rconnected}, $q=2$), $d\in L(e)$
because it is an endpoint of $f\prec_ye$.  By Lem.~\ref{lem:p++}, $q^{-+}\cap
d^{+-}=\emptyset$, hence $d_y<q_y$ and $R(f)\subset A(c,d)<_yq_{-+}$.  If
$q_x=d_x<b_x$ (Fig.~\ref{fig:Rconnected}, $q=3$), the same reasoning applies.
If $q_x=d_x=b_x$ (Fig.~\ref{fig:Rconnected}, $q=4$), $q_y\geq b_y$ and $b_y\geq
d_y$ by constraint~2, so $R(f)\subset A(c,d)<_yq^{-+}$.
\end{proof}
\begin{figure}[tbp]
\centering \includegraphics[width=4in]{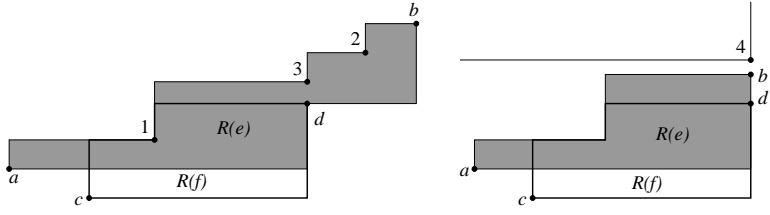}
\caption{$R(e)$ is shaded and $R(f)$ has a thicker
  boundary.}\label{fig:Rconnected}
\end{figure}

\begin{theorem}\label{t-iff}
An approximate subdivision has a realization iff it is consistent.
\end{theorem}

\begin{proof}
Lem.~\ref{lem:onlyif} proves necessity.  Lems.~\ref{lem:ordered},
\ref{lem:Rmonotone}, and~\ref{lem:Rordered} prove sufficiency.
\end{proof}

\section{Error Analysis}\label{s-error}

We prove that a consistent $\delta$-splitting of a $\delta$-accurate approximate
subdivision is a $\delta$-refinement (Thm.~\ref{t-refine}).  This implies that
the refinement algorithm output is a $\delta$-refinement (Thm.~\ref{t-delta}).
We first prove the special case where the approximate subdivision is consistent,
hence is its own refinement (Thm.~\ref{t-realize}).  We define the
$\delta$-offsets of a curve (Def.~\ref{def:offset} and Fig.~\ref{f-offset}a) and
prove them monotone (Lem.~\ref{l-offsetmon}).  We prove that the $\prec_y$
relation on the curves of a $\delta$-accurate approximate subdivision implies a
$<_y$ relation on their $\delta$-offsets
(Lems.~\ref{l-offset}--\ref{l-poffset}).  We use these tools to sharpen the
Thm.~\ref{t-iff} realization into one that is $\delta$-close to the approximate
subdivision.

\begin{lemma}\label{l-distmon}
For a point, $a$, and an increasing (decreasing) curve, $e$, $\dist(a, e)$ is
decreasing (increasing) in $a_x$ and increasing in $a_y$ for $a>_ye$ and is
increasing (decreasing) in $a_x$ and decreasing in $a_y$ for $a<_ye$.
\end{lemma}
\begin{proof}
Follows from Lem.~8 in our prior paper.\cite{sacks-milenkovic06}
\end{proof}

\begin{definition}[$e_{-\delta}$ and $e_{+\delta}$]\label{def:offset}
The {\em $\delta$-offsets\/} of a curve, $e$, are
\begin{eqnarray*}
e_{-\delta} & = & \{a<_ye| \dist(a,e)=\delta\}\\
e_{+\delta} & = & \{a>_ye| \dist(a,e)=\delta\}.
\end{eqnarray*}
\end{definition}

\begin{lemma}\label{l-offsetmon}
The offsets $e_{\pm\delta}$ are monotone.
\end{lemma}
\begin{proof}
Let a vertical line intersect $e_{\pm\delta}$ at $a$ and $b$ with $a_y<b_y$.
Since $\dist(a,e)=\delta$ and $\dist(b,e)=\delta$, $\dist(c,e)=\delta$ for
$c\in[a,b]$ by Lem.~\ref{l-distmon}.  Hence, the line intersects $e_{\pm\delta}$
in a connected set.  Likewise for a horizontal line.
\end{proof}

Lems.~\ref{l-offset} and~\ref{l-poffset} apply to a $\delta$-accurate
approximate subdivision.

\begin{lemma}\label{l-offset}
For two curves, $e$ and $f$, $e\preceq_yf$ implies $e_{-\delta}<_yf_{+\delta}$.
\end{lemma}
\begin{proof}
By Def.~\ref{def:offset}, $e_{-\delta}<_y e <_y e_{+\delta}$.  For $e\prec_y f$,
let $a\in e_{-\delta}$ and $b\in f_{+\delta}$ with $a_x=b_x$
(Fig.~\ref{f-offset}b).  There exist $\delta$-deformations, $\tilde{e}$ and
$\tilde{f}$, with $\tilde{e}<_y\tilde{f}$.  Let $\tilde{a}\in\tilde{e}$ and
$\tilde{b}\in\tilde{f}$ satisfy $\tilde{a}_x=\tilde{b}_x=a_x$, so
$\tilde{a}_y<\tilde{b}_y$.  We have $a_y<\tilde{a}_y$ by Lem.~\ref{l-distmon}
because $\dist(a,e)=\delta$ and $\dist(\tilde{a},e)<\delta$; likewise
$\tilde{b}_y<b_y$.  Thus, $a_y<\tilde{a}_y<\tilde{b}_y<b_y$.
\end{proof}

\begin{figure}[tbp]
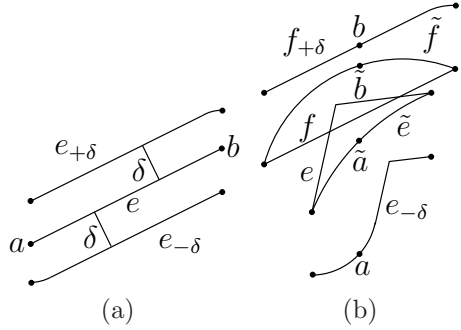

\centering
\begin{tabular}{cc}
\input{offset1L.pstex_t} &
\input{offset2L.pstex_t} \\
(a) & (b)
\end{tabular}
\caption{Offsets (a) and their properties (b).}\label{f-offset}
\end{figure}

\begin{lemma}\label{l-poffset}
If $p\in L(e)$, $p<_y\overline{e_{+\delta}}$; if $q\in U(e)$,
$q>_y\overline{e_{-\delta}}$.
\end{lemma}
\begin{proof}
We prove the $p \in L(e)$ case for $e[a,b]$.  By construction, $e<_ye_{+\delta}$
hence $\{a,b\}<_y\overline{e_{+\delta}}$ by continuity and because
$\dist(a,e)=\dist(b,e)=0<\delta$.  If $\tilde{e}$ is a $\delta$-deformation of
$e$, $\overline{\tilde{e}}=\tilde{e}\cup\{a,b\}<_y\overline{e_{+\delta}}$ and
similarly $\overline{\tilde{e}}>_y\overline{e_{-\delta}}$.

If $p\in L(e)$, $p$ is an endpoint of $f\prec_ye$ and hence an endpoint of
$\tilde{f}<_y\tilde{e}$ by Def.~\ref{def:delta-accurate}.  By continuity,
$p_y\leq r_y$ for $r\in\overline{\tilde{e}}$ with $p_x=r_x$.  Since
$\overline{\tilde{e}}<_y\overline{e_{+\delta}}$, $r<_y\overline{e_{+\delta}}$
hence $p<_y\overline{e_{+\delta}}$.
\end{proof}

For the proof of Thm.~\ref{t-realize}, $R_\delta(e)$ serves the role
that $R(e)$ did for Thm.~\ref{t-iff}.

\begin{definition}[$R_\delta(e)$]\label{def:Rdelta}
For an approximate subdivision, define $R_\delta(e) = R(e)$ for horizontal and
vertical $e$.  For increasing $e$, $R_\delta(e)$ is $R(e)$ minus $f_{-\delta}^-$
for all increasing $f\preceq_ye$ and minus $g_{+\delta}^+$ for all increasing
$g\succeq_y e$ (Fig.~\ref{fig:Rdelta}).  Similarly for decreasing $e$.
\end{definition}
\begin{figure}[tbp]
\centering
\includegraphics[width=4in]{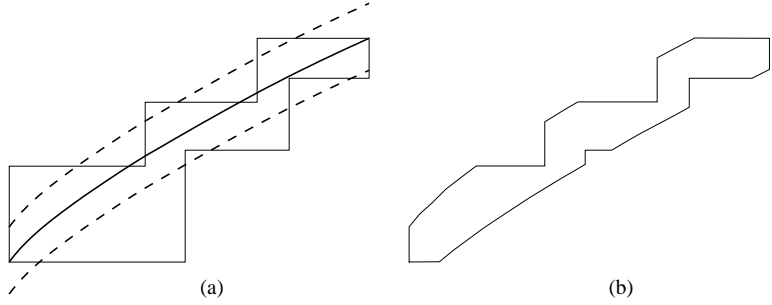}
\caption{a) $e$, $e_{-\delta}$, $e_{+\delta}$, and $R(e)$; (b)
  $R(e)-e_{-\delta}^-\cup e_{+\delta}^+$.}\label{fig:Rdelta} 
\end{figure}

\begin{lemma}\label{lem:Rdeltapath}
If an approximate subdivision is consistent and $\delta$-accurate, the sets
$R_\delta(e)$ are monotone.
\end{lemma}
\begin{proof}
We prove the $e[a,b]$ increasing case.  Let $f[c,d]\preceq_ye$ be increasing.
For $q\in U(f)$, $\overline{f_{-\delta}}<_yq$ by Lem.~\ref{l-poffset}, so
$f_{-\delta}<_yq^{-+}$, since $f_{-\delta}$ is increasing.  Hence,
$f_{-\delta}<_yU(f)^{-+}$, which implies $f_{-\delta}\subset R(f)^-$, since both
span the $x$-interval $(c_x,d_x)$.  Since $R(f)<_yu(e)$, $R(f)^-<_yu(e)$ and
$f_{-\delta}<_yu(e)$.  Likewise, $g_{+\delta}>_yl(e)$ for increasing
$g\succeq_ye$.  If $f\preceq_y e$, $e\preceq_yg$, and $f$ and $g$ overlap in
$x$, $f\preceq_yg$ by Lem.~\ref{l-trans} and $f_{-\delta}<_yg_{+\delta}$ by
Lem.~\ref{l-offset}.  We conclude that the lower and upper sets that are removed
from $R(e)$ are disjoint.  Furthermore, $c,d\not\in R(e)$ by
Lem.~\ref{lem:Rordered} and $f_{-\delta}<_yf$.  Thus, the only part of the
boundary of $f_{-\delta}^-$ that enters $R(e)$ is part of $f_{-\delta}$, which
is monotonic.  Therefore, the lower boundary of $R_\delta(e)$ is monotonic and
similarly the upper boundary.
\end{proof}

\begin{lemma}\label{lem:Rdeltaord}
Let $u_\delta(e)$ be the upper boundary of $R_\delta(e)$.  If
$f[c,d]\prec_ye[a,b]$ are both increasing or both decreasing,
$R_\delta(f)<_yu_\delta(e)$.
\end{lemma}
\begin{proof}
Pick $p\in R_\delta(f)$ with $a_x\leq p_x\leq b_x$.  Since $R_\delta(f)\subset
R(f)<_y u(e)$, $p<_y u(e)$.  Pick $g\succeq_y e$ that overlaps $p$ in $x$.
Thus, $g$ overlaps $f$ in $x$ and $g\succ_y f$ by Lem.~\ref{l-trans}.
Therefore, $g_{+\delta}^+$ was removed from $R_\delta(f)$ hence
$p<_yg_{+\delta}^+$.  Since $p$ is below the upper sets removed from
$R_\delta(e)$ that overlap it in $x$, $p<_yu_\delta(e)$.  Hence,
$R_\delta(f)<_yu_\delta(e)$.
\end{proof}

\begin{theorem}\label{t-realize}
A consistent, $\delta$-accurate approximate subdivision is a $\delta$-refinement.
\end{theorem}
\begin{proof}
The $R_\delta(e)$ are monotone by Lem.~\ref{lem:Rdeltapath}.  They are an
ordered assignment because $R_\delta(e)\subset R(e)$, so
$R_\delta(e)<_yR_\delta(f)$ for $e\prec_yf$ unless both are decreasing or both
are increasing, and that case is covered by Lem.~\ref{lem:Rdeltaord}.  By
Lem.~\ref{lem:ordered}, there is a realization, $r(e)\subset R_\delta(e)$.  For
$p\in R_\delta(e)$, $\dist(p,e)<\delta$ by construction, so $r(e)$ is a
$\delta$-deformation of $e$.
\end{proof}

We prove Thm.~\ref{t-refine} similarly, using Lem.~\ref{l-accemb} instead of
Lem.~\ref{l-poffset} and $R'_\delta(e')$ (Def.~\ref{R'delta}) instead of
$R_\delta(e)$.

\begin{lemma}\label{l-accemb}
In a $\delta$-splitting, $S'$, of a $\delta$-accurate approximate subdivision,
$S$, for every endpoint $p\in S$, curve $e\in S$, and curve $e'\in C(e)$, $p\in
L(e')$ implies $p<_y\overline{e_{+\delta}}$ and $p\in U(e')$ implies
$p>_y\overline{e_{-\delta}}$.
\end{lemma}
\begin{proof}
We prove the $p\in L(e')$ case.  Consider $e'[a',b']$.  If $p\in L(e)$,
$p<_y\overline{e_{+\delta}}$ by Lem.~\ref{l-poffset}.  If $p\in U(e)$, $p$ is an
endpoint of $f\in S$ with $e\prec_y f$.  Thus, $p$ is an endpoint of a curve,
$f'\in C(f)$, but by Def.~\ref{def:refappsub}, $e'\preceq_y f'$, so $p\in
U(e')$.  By Lem.~\ref{lem:p++}, $p\in\{a',b'\}$ hence $\dist(p,e)<\delta$ by
Def.~\ref{def:accemb} and so $p<_y\overline{e_{+\delta}}$.
\end{proof}

\begin{definition}[$e'_{-\delta}$, $e'_{+\delta}$, $R'_\delta(e')$]\label{R'delta}
For increasing $e'[a',b']\in C(e)$, define $e'_{-\delta}$ and $e'_{+\delta}$ to
be $e_{-\delta}$ and $e_{+\delta}$ restricted to $(a_x,b_x)$ if $e$ is
increasing; otherwise $e'_{-\delta}=e'_{+\delta}=\emptyset$.  Similarly for
decreasing $e'$.  Define $R'_\delta(e')=R(e')$ for horizontal or vertical $e'$.
For increasing $e'$, $R'_\delta(e')$ is $R(e')$ minus ${f'}_{-\delta}^-$ for
increasing $f'\preceq_ye'$ and minus ${g'}_{+\delta}^+$ for increasing
$g'\succeq_ye'$, and similarly for decreasing $e'$.
\end{definition}

\begin{theorem}\label{t-refine}
A consistent $\delta$-splitting of a $\delta$-accurate approximate subdivision
is a $\delta$-refinement.
\end{theorem}
\begin{proof}
Modify the proof of Thm.~\ref{t-realize} by replacing Lem.~\ref{l-poffset} with
Lem.~\ref{l-accemb} to select a realization $r(e')\subset R'_\delta(e')$ for
each curve $e'[a',b']$.  If $e'$ and $e$ are both increasing or both decreasing,
$R'_\delta(e')$ lies between $e_{-\delta}$ and $e_{+\delta}$ by construction
hence $r(e')$ lies within $\delta$ of $e$.  If $e$ is increasing but $e'$
decreasing, we know $\dist(a',e)<\delta$ hence $a'<_ye_{+\delta}$ hence
${a'}^{+-}<_ye_{+\delta}$.  Similarly, ${b'}^{-+}>_ye_{-\delta}$.  Hence,
$r(e')\subset {a'}^{+-}\cap{b'}^{-+}$ lies within $\delta$ of $e$.
\end{proof}

\begin{lemma}\label{l-refine}
The refinement algorithm output is a $\delta$-splitting.
\end{lemma}
\begin{proof}
Let the algorithm split a curve, $e$, at an endpoint, $p$.  We show that
$\dist(p,e)<\delta$.

Step~1: Necessarily, $p\in L(e)$, because $p$ belongs to a group below $e$, and
there is a violation of constraint~1.  Hence, there exists $q\in U(e)$ with
$p_x=q_x$ and $p_y\geq q_y$.  By Lem.~\ref{l-poffset}, $p<_ye_{+\delta}$ and
$q>_ye_{-\delta}$.  Since $p_y\geq q_y$, $p>_ye_{-\delta}$ as well.

Step~2: For an incoming curve, $e[v,a]$, define $\M(e)$ analogously to $\m(e)$
as the maximum safe endpoint over $f[w,b]\preceq_ye$.  Since $e\preceq_ye$ and
$\dist(e,a)\leq\dist(e,a)$, $\m(e)\leq_ya\leq_y\M(e)$.  Curve $e$ is split at
the maximum $\m(f)$ with $f\preceq_ye$.  Let $f[w,b]$ achieve this maximum.
Since $e\preceq_ye$, $\m(e)\leq_y\m(f)$.  If $b$ is safe for $e$
($\dist(b,e)\leq\dist(f,a)$), $\m(f)\leq_yb$ from above and $b\leq_y\M(e)$ by
definition of $\M(e)$; otherwise, if $a$ is safe for $f$, $\m(f)\leq_ya$ by
definition of $\m(e)$ and $a\leq_y\M(e)$ from above.  Hence $\m(f)\leq_yM(e)$.
Since both $\m(e)$ and $\M(e)$ are safe for $e$,
$\dist(\m(e),e),\dist(\M(e),e)<\delta$ by Lem.~\ref{l-safedist} below.  The
error bound neglects the distance computation error, $\gamma$, which is
typically negligible with respect to $\delta$.  Since we evaluate $\dist$
approximately, we might pick the endpoint that is unsafe by a little bit.
Strictly speaking, we should replace $\delta$ with $\delta+2\gamma$.

Step~3: We consider constraint~3 with $e$ increasing; the other cases are
similar.  Since $p\in L(e)$ and $d\in U(e)$, $p<_ye_{+\delta}$, and
$d>_ye_{-\delta}$ by Lem.~\ref{l-poffset}.  Also, $p>_ye_{-\delta}$ because
$p_x<d_x$, $p_y\geq d_y$, and $e_{-\delta}$ is increasing by
Lem.~\ref{l-offsetmon}.

Step~4: $\dist(p,e)=0$.
\end{proof}

\begin{lemma}\label{l-safedist}
If $e[v,a]\prec_yf[w,b]$ violate constraint~2 at $a_x=b_x$, $\dist(a,f)<\delta$
or $\dist(b,e)<\delta$.
\end{lemma}
\begin{proof}
Let $\tilde{e}$ and $\tilde{f}$ be curves that show $\delta$-accuracy.  Since
$\tilde{e}<_y\tilde{f}$ and $a_y>b_y$, either $\tilde{e}$ has a vertical segment
at $x=a_x$ that contains $b$ or $\tilde{f}$ has such a segment that contains
$a$.  The result follows because $\tilde{e}$ and $\tilde{f}$ are $\delta$-close
to $e$ and $f$.
\end{proof}

\begin{theorem}\label{t-delta}
The refinement algorithm computes a $\delta$-refinement.
\end{theorem}
\begin{proof}
Follows from Thm.~\ref{t-refinealg}, Thm.~\ref{t-refine}, and Lem.~\ref{l-refine}.
\end{proof}

\section{Set Operations}\label{s-set}  

We now turn to set operations.  We compute the overlay approximate subdivision
with our sweep algorithm\cite{sacks-milenkovic06} then refine.  The sweep
algorithm computes the curve crossing points approximately and preserves the
curve endpoints, so the output endpoints are the union of the input endpoints
and the approximate crossing points.  The overlay $\prec_y$ determines the
membership of the overlay faces in the input faces via a standard
algorithm.\cite{deberg08}

We modify the sweep algorithm to use red/blue sweep list insertion.  When
inserting a curve from the first input (red), find the two nearest red curves,
using the red $\prec_y$, then compute the order of the new curve with respect to
the intervening blue curves.  Likewise for blue curve insertion.  We avoid
computation of red/red and blue/blue orders.  We also ensure that the overlay
$\prec_y$ agrees with the red and blue $\prec_y$, so overlay face membership is
computed correctly.

\section{Euclidean Transformations}\label{s-transform}

We transform an approximate subdivision by transforming each curve, computing
the transformed $\prec_y$, and invoking the refinement algorithm.  The
computational complexity is dominated by the refinement algorithm.

Translating by $t$ maps $e$ to $\{p+t|p\in e\}$ and drops each curve whose
endpoints translate to the same point due to rounding.  It does not change
$\prec_y$ or $\delta$-accuracy. Scaling ($y$) by $w$ maps $e$ to
$\{(p_x,wp_y)|p\in e\}$ and drops the curves with equal endpoints.  It does not
change $\prec_y$, but scales $\delta$-accuracy by $w$.  Skewing by $w$ maps $e$
to $\{(p_x,wp_x+p_y)|p\in e\}$ and splits the skewed curves at their $y$ turning
points.  It increases the distance between two points by at most a factor of
$z=\sqrt{w^2+1}$.  Skewing a deformation yields a deformation.  Hence, skewing
scales $\delta$-accuracy by at most $z$.

\paragraph{Rotation}

Rotating by $90^\circ$ maps $e$ to $\{(-p_y,p_x)|p\in e\}$.  The rotated curves
are monotone because the input curves are.  The rotated $\prec_y$ is the
horizontal order of the input curves.  Let the rotations of $e$ and $f$ be $e^r$
and $f^r$: $e^r\prec_y f^r$ iff $e$ is left of $f$.  It is $\delta$-accurate
because a deformation that places two curves in $\prec_y$ order also places them
in horizontal order.

Compute the horizontal order via a horizontal line sweep whose events are
insertion and removal of each curve at its lower and upper endpoints.  Removal is
standard.  A curve, $f$, is inserted by computing the order of its lower
endpoint, $v$, with respect to $\bigo(\log n)$ curves in the sweep list.  Let
$e[t,h]$ be such a curve.
\begin{enumerate}
\item If $v_x<t_x$, $v$ is left of $e$.
\item If $f\prec_y e$ and $e$ is increasing, $v$ is right of $e$;
      if $e$ is decreasing, $v$ is left of $e$.
\item If $e\prec_y f$ and $e$ is increasing, $v$ is left of $e$;
      if $e$ is decreasing, $v$ is right of $e$.
\item If $h_x<v_x$, $v$ is right of $e$.
\end{enumerate}
Fig.~\ref{f-hsweep} shows points $v_1,\ldots,v_4$ that satisfy the four rules.
			
\begin{figure}[tbp]
\centering
\begin{tabular}{cc}
\includegraphics{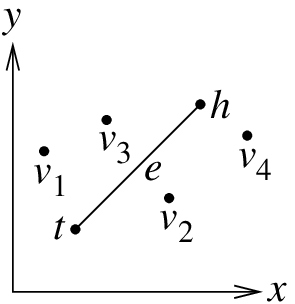} & \includegraphics{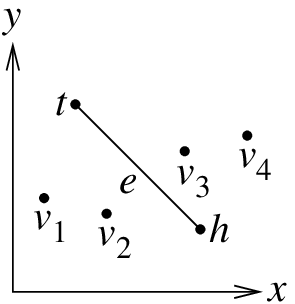}\\
(a) & (b)
\end{tabular}
\caption{Horizontal endpoint/$\prec_y$ for increasing (a) and decreasing (b)
curves.}
\label{f-hsweep}
\end{figure}

Rotating by $-90^\circ$ maps $e$ to $(e_y,-e_x)$.  The rotated $\prec_y$ is the
opposite of the horizontal order of the input curves.

Rotation by $\theta$ is performed in two steps.  Step~1 rotates the input by
$90^\circ$ and replaces $\theta$ with $\theta-90^\circ$ for
$45^\circ<\theta<135$, rotates the input by $-90^\circ$ and replaces $\theta$
with $\theta+90^\circ$ for $-135^\circ<\theta<-45^\circ$, and is omitted for
other $\theta$ values.  Let $s=\sin\theta$ and $c=\cos\theta$.  Step~2 scales by
$c$, skews by $s$, rotates by $90^\circ$, scales by $1/c$, skews by $s/c$, and
rotates by $-90^\circ$.  The rotation formula is verified by direct calculation.
The $\delta$-accuracy scales by
\begin{displaymath}
1\times c \times \sqrt{1+s^2}\times 1\times \frac{1}{c}\times
\sqrt{1+\frac{s^2}{c^2}}\times1=\sqrt{\frac{1+s^2}{c^2}}.
\end{displaymath}
This factor is bounded by $\sqrt{3}$ because $|\theta|\leq 45^\circ$ after
step~1.

\section{Validation}\label{s-validate}

We have implemented set operations and Euclidean transformations for approximate
subdivisions whose curves are approximate algebraic-curve segments.  A curve is
represented by its endpoints, $t$ and $h$ with $t_x\leq h_x$, and by a monotone
branch, $b$, of an algebraic curve.  The curve is the portion of $b$ in the box
with corners $t$ and $h$, plus two line segments that link $t$ and $h$ to it
(Fig.~\ref{f-curve}).  If $t_y<h_y$ and $b$ is decreasing or $h_y<t_y$ and $b$
is increasing, the curve is the line segment $th$.  When a curve is split, the
new curves inherit its branch.  Our numerical solver\cite{sacks-milenkovic06}
computes points on curves by solving univariate polynomials and computes curve
crossing points by solving pairs of bivariate polynomials.

\begin{figure}[tbp]
\centering
\begin{tabular}{ccc}
\includegraphics{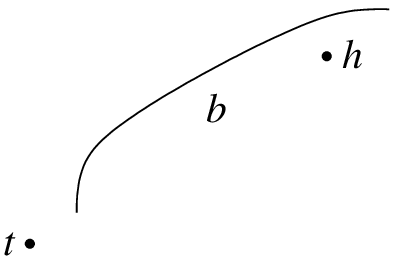} &
\includegraphics{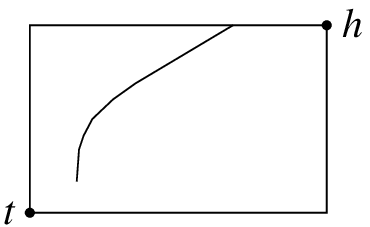} &
\includegraphics{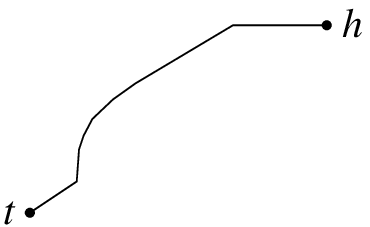}\\
(a) & (b) & (c)
\end{tabular}
\caption{Curve representation (a), portion of $b$ in box (b), and curve (c).}
\label{f-curve}
\end{figure}

We validate the algorithms on sequences of set operations alternating with
Euclidean transformations.  We test the four shapes shown in
Fig.~\ref{f-validate}.  The shapes are drawn in the box $[-1,1]\times[-1,1]$.
The curves are line segments in shape~1, line and circle segments in shape~2,
degree-3 curves in shape~3, and degree-6 curves in shape~4.  We initialize $s_0$
to a shape then set $s_{i+1}$ to the symmetric difference of $s_i$ with a
Euclidean transform of $s_i$ for six iterations (e--h).  The $s_0$ transform is
translation by $(0.0123, 0.0321)$ and rotation by $0.765$ radians.  These values
are halved after each iteration.

\begin{figure}[tbp]
\centering
\begin{tabular}{cccc}
\includegraphics[width=1in]{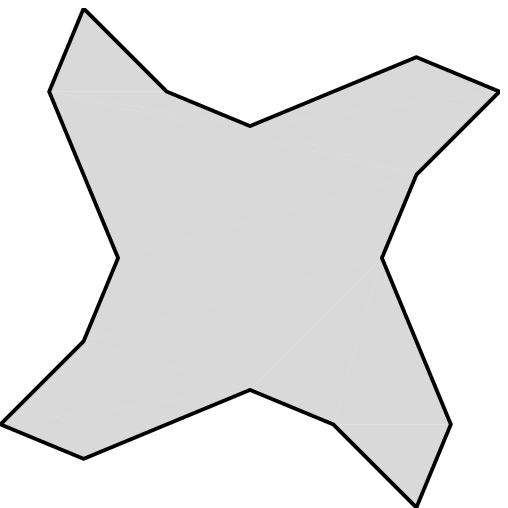} &
\includegraphics[width=1in]{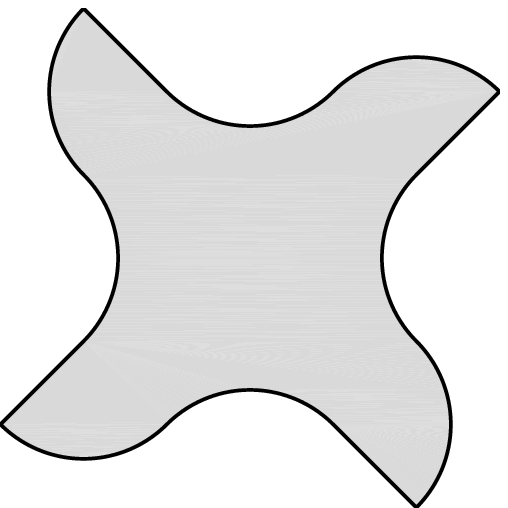} &
\includegraphics[width=1in]{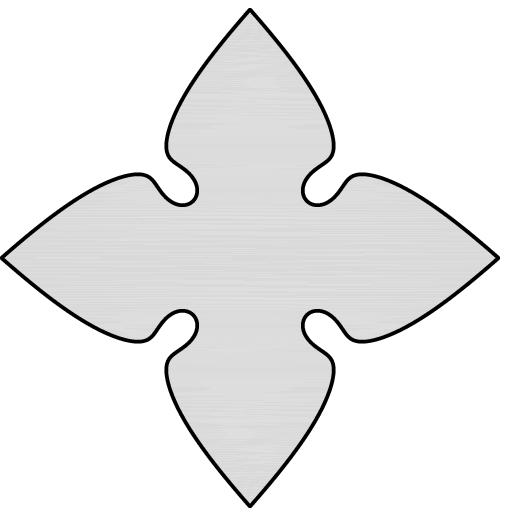} &
\includegraphics[width=1in]{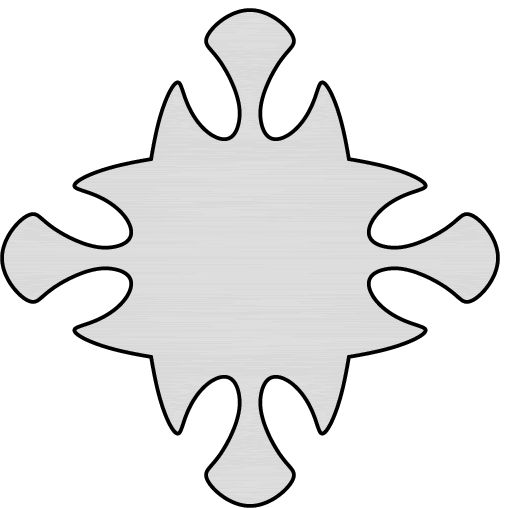}\\
(a) & (b) & (c) & (d)\\
\includegraphics[width=1in]{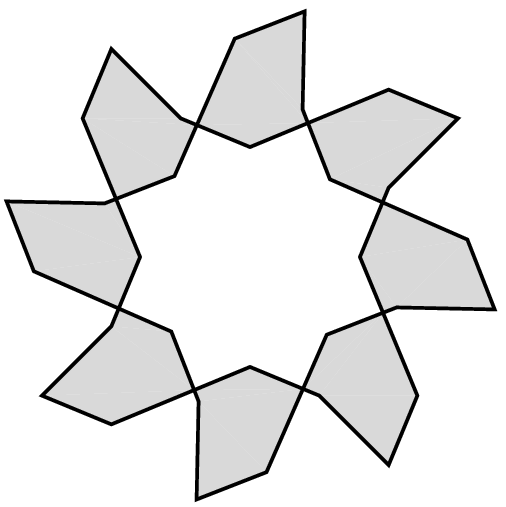} &
\includegraphics[width=1in]{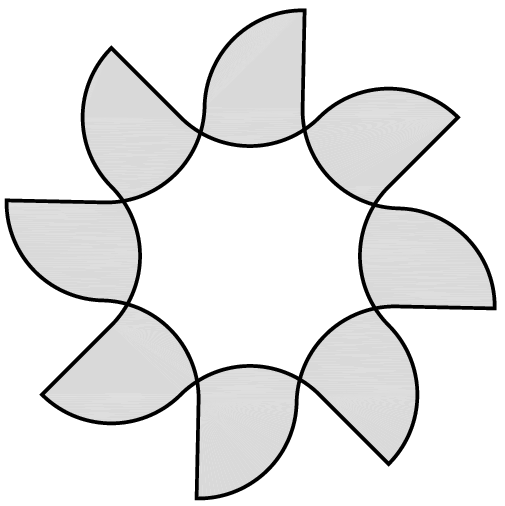} &
\includegraphics[width=1in]{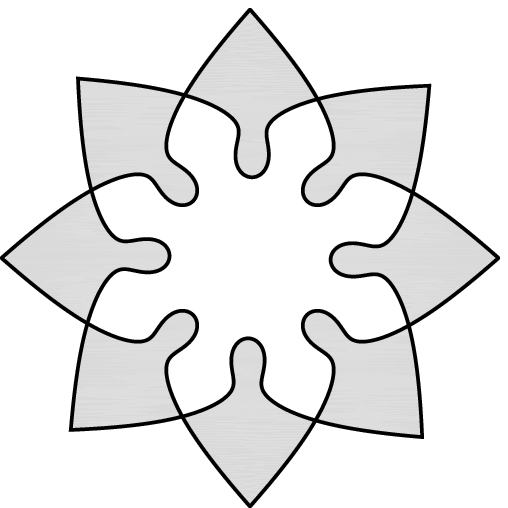} &
\includegraphics[width=1in]{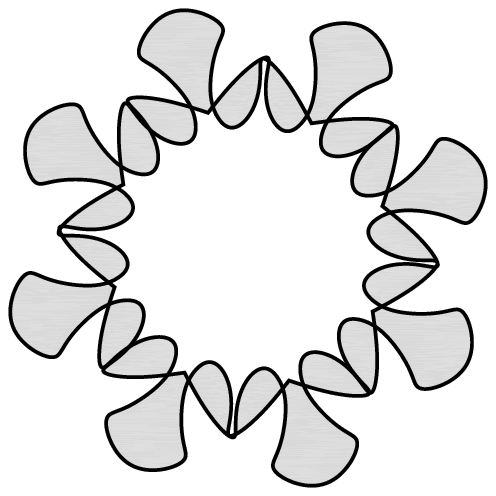}\\
(e) & (f) & (g) & h
\end{tabular}
\caption{Validation: shapes 1--4 (a--d) and their $s_1$ regions (e--h).}
\label{f-validate}
\end{figure}

The initial tests have high combinatorial complexity and no degeneracy.  The
program encounters no constraint violations, generates accurate outputs, and is
fast.  Table~\ref{t-res} summarizes the results.  The input size, $n$, is the
number of input curves.  The output size, $e$, is the number of curves in the
final, $s_6$, subdivision.  The output error, $\delta$, is the larger of two
error metrics both of which are the maximum over the six iterations.  The first
metric is the distance between the monotone branch of a curve and one of its
endpoints.  The second metric is the distance between a curve, $e$, and an
endpoint, $v$, with $v\in L(e)$ and $v>_ye$, or $v\in U(e)$ and $v<_ye$.  The
running time, $t$, is for one core of an Intel Core~2 Duo.  As the algebraic
degree of the input increases from one to six, the constraint enforcement
percentage, $c$, decreases rapidly and the numerical solver percentage, $s$,
increases rapidly.

\begin{table}[tbp]
\caption{Initial tests: $n$ is the input size, $e$ is the output size, $\delta$
  is the output error, $t$ is the running time in seconds, $c$ is the constraint
  enforcement percentage, and $s$ is the solver percentage.}\label{t-res}
\centering
\begin{tabular}{r|cccccc}
shape & $n$ & $e$    & $\delta$           & $t$  & $c$ & $s$\\ \hline
1     & 16  & 67000  & $3\times10^{-16}$  & 1.1  & 23  & 20\\
2     & 20  & 70000  & $6\times10^{-16}$  & 3.6  & 10  & 70\\
3     & 40  & 85000  & $8\times10^{-15}$  & 4.5  & 8   & 67\\
4     & 80  & 277000 & $4\times10^{-12}$  & 127  & 1   & 91
\end{tabular}
\end{table}

We repeated the tests with the $s_0$ Euclidean transformation divided by $2^j$
for $j=1,\ldots,50$.  As the transformation shrinks, the six iterates of the
input shape converge, so the symmetric differences approach degeneracy.
Fig.~\ref{f-res} plots the results for the initial, $j=0$, test and for the 50
subsequent tests.  The output error and the solver percentage are omitted
because they are essentially constant.

\begin{figure}[tbp]
\centering
\begin{tabular}{cc}
\includegraphics{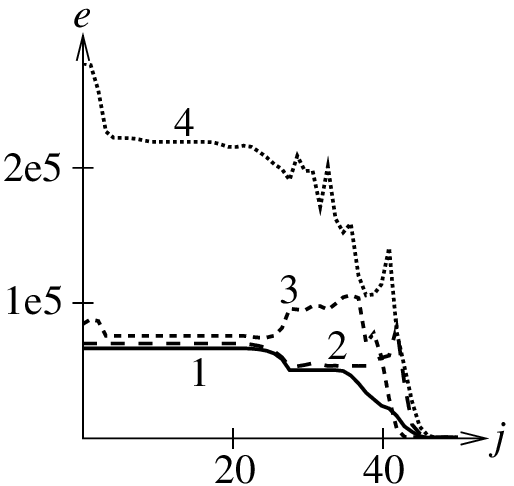} & \includegraphics{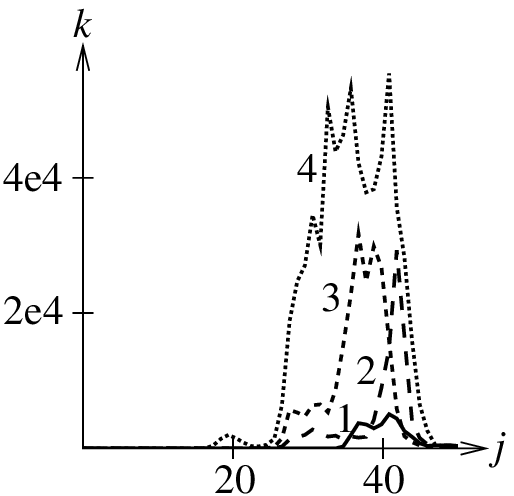}\\
\includegraphics{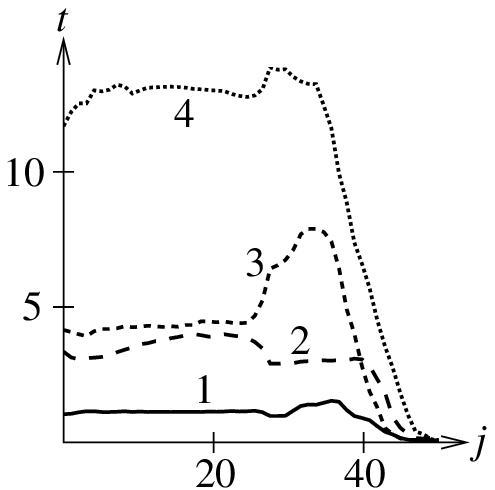} & \includegraphics{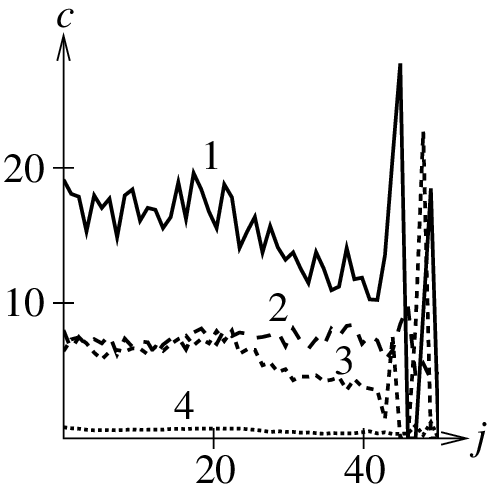}\\
\end{tabular}
\caption{All tests: $j$ is the iteration, $e$ is the output size, $k$ is the
  number of constraint violations, $t$ is the running time, and $c$ is the
  constraint enforcement percentage.}\label{f-res}
\end{figure}

The output size decreases because each iterate has fewer intersections with the
previous iterate.  The number of constraint violations is zero for $j<15$,
increases sharply from $j=25$ to $j=40$, then decreases to nearly zero.  We
expect that the curves become nearly identical rather than intersecting.  The
shape~3 running time increases from $j=24$, peaks at double the $j=0$ time at
$j=32$, and decreases to nearly zero.  The running time for the other shapes
increases slightly if at all then decreases to nearly zero.  The running time
for shape~4 is divided by ten, so it can be displayed with the others.  The
percentage spent on constraint enforcement decreases, except for a peak in
$40<j<50$ where the output is small relative to the number of violations.  The
percentage also decreases as the input degree increases: constraint enforcement,
which is combinatorial, is independent of degree, whereas curve intersection,
which is numerical, is polynomial in degree.

\section{Conclusions}\label{s-conclude}

This paper presents approximate planar shape-manipulation algorithms.  The
validation results show that the algorithms are accurate and fast on both
generic and degenerate inputs.  Constraint enforcement takes a small fraction of
the running time, whereas numerical computation takes most of the time.  The
results support our thesis that inconsistency sensitive algorithms are efficient
and accurate.

\section*{Acknowledgments}

Research supported by NSF grants IIS-0082339, CCF-0306214, and CCF-0304955.

\bibliographystyle{unsrt}
\bibliography{bib,bib2}
\end{document}